\documentclass[journal]{IEEEtran}   
\usepackage{multirow}
\usepackage{diagbox}
\usepackage{amssymb}
\usepackage{amsmath}
\usepackage{graphicx}
\usepackage{cite}
\usepackage{citesort}
\usepackage{subfigure}
\usepackage{graphicx,epstopdf}
\usepackage{epsfig}	
\usepackage{cite,graphicx,amsmath,amssymb}
\usepackage{comment}

\usepackage{amssymb}
\usepackage{amsmath}
\usepackage{cite}
\usepackage{url}
\usepackage{xcolor}
\usepackage{cite,graphicx,amsmath,amssymb}
\usepackage{subfigure}
\usepackage{citesort}
\usepackage{fancyhdr}
\usepackage{mdwmath}
\usepackage{mdwtab}
\usepackage{caption}
\usepackage{amsthm}
\usepackage{lipsum}

\usepackage [latin1]{inputenc}

\newtheorem{theorem}{Theorem}

\newtheorem{lemma}{Lemma}

\newtheorem{corollary}{Corollary}

\newtheorem{remark}{Remark}  

\makeatletter
\def\ScaleIfNeeded{%
	\ifdim\Gin@nat@width>\linewidth \linewidth \else \Gin@nat@width
	\fi } \makeatother

\begin{document}
	
\title{Exploiting Active RIS in NOMA Networks with Hardware Impairments}

\author{Xinwei~Yue,~\IEEEmembership{Senior Member,~IEEE}, Meiqi Song, Chongjun Ouyang, Yuanwei\ Liu,~\IEEEmembership{Fellow,~IEEE}, Tian Li,~\IEEEmembership{Member,~IEEE} and Tianwei Hou,~\IEEEmembership{Member,~IEEE}
	
	\thanks{X. Yue and M. Song are with the Key Laboratory of Information and Communication Systems, Ministry of Information Industry and also with the Key Laboratory of Modern Measurement $\&$ Control Technology, Ministry of Education, Beijing Information Science and Technology University, Beijing 100101, China (email: \{xinwei.yue and meiqi.song\}@bistu.edu.cn).}
	\thanks{C. Ouyang is with the School of Electrical and Electronic Engineering, University College Dublin, Dublin, D04 V1W8, Ireland (e-mail: chongjun.ouyang@ucd.ie).}
	\thanks{Y. Liu is with the School of Electronic Engineering and Computer Science, Queen Mary University of London, London E1 4NS, U.K. (email: yuanwei.liu@qmul.ac.uk).}
	\thanks{T. Li is with the 54th Research Institute of China Electronics Technology Group Corporation, Shijiazhuang Hebei 050081, China. (email: t.li@ieee.org).}
	\thanks{T. Hou is with the School of Electronic and Information Engineering, Beijing Jiaotong University, Beijing 100044, China, and also with the Institute for Digital Communications, Friedrich-Alexander Universit{\"a}t Erlangen-N{\"u}rnberg (FAU), 91054 Erlangen, Germany (email: twhou@bjtu.edu.cn). }
}
	
\maketitle
	
\begin{abstract}
	Active reconfigurable intelligent surface (ARIS) is a promising way to compensate for multiplicative fading attenuation by amplifying and reflecting event signals to selected users. This paper investigates the performance of ARIS assisted non-orthogonal multiple access (NOMA) networks over cascaded Nakagami-$m$ fading channels. The effects of hardware impairments (HIS) and reflection coefficients on ARIS-NOMA networks with imperfect successive interference cancellation (ipSIC) and perfect successive interference cancellation (pSIC) are considered.
	More specifically, we develop new precise and asymptotic expressions of outage probability and ergodic data rate with ipSIC/pSIC for ARIS-NOMA-HIS networks. According to the approximated analyses, the diversity orders and multiplexing gains for couple of non-orthogonal users are attained in detail. Additionally, the energy efficiency of ARIS-NOMA-HIS networks is surveyed in delay-limited and delay-tolerant transmission schemes.
	The simulation findings are presented to demonstrate that: i) The outage behaviors and ergodic data rates of ARIS-NOMA-HIS networks precede that of ARIS aided orthogonal multiple access (OMA) and passive reconfigurable intelligent surface (PRIS) aided OMA; ii) As the reflection coefficient of ARIS increases, ARIS-NOMA-HIS networks have the ability to provide the strengthened outage performance; and iii) ARIS-NOMA-HIS networks are more energy efficient than ARIS/PRIS-OMA networks and conventional cooperative schemes.
\end{abstract}
	
	
\section{Introduction}
As the diversification of business demand increases, wireless communications are currently confronted with unpredictable challenges \cite{Chenshanzhi6G2020,Nguyen6GIOT2022}. It is forecasted that ten thousand mobile users and smart equipment will be linked to the networks by the end of 2022 \cite{Cisco2019}.
The aims of the sixth-generation (6G) networks are predicted to support the key performance indexes, i.e., higher data rate, massive connections, low latency, as well as higher full coverage, and so on. Non-orthogonal multiple access (NOMA) has been referred to as one of the physical layer technologies \cite{Yuanwei2022NOMANGMA,Xinyue20226GNOMA}, which solves the problems of limited access for multiple users from the views of information-theoretic.
Furthermore, the integration of NOMA with emerging wireless techniques was explored towards 6G networks \cite{Yuanwei2022NOMA}. Another noteworthy technique, i.e., reconfigurable intelligent surface (RIS) has sparked the enthusiasm of research from both academic and industrial communities \cite{Cui2014metamaterials,Hum2014metamaterials}. It has been confirmed to accomplish the control of incident electromagnetic waves through the programmable device.
	
In particular, the RIS is generally arranged by a mount of passive electromagnetic units designed carefully, which can adapt to changes in the environment \cite{WuTowards2020PRIS,Basharat2021PRIS6G,Yuanwei2021PRIS}.
Up to now, numerous treatises have introduced this type of passive RIS (PRIS) into wireless communication networks. In \cite{Huang2019PRIS}, a new design of energy efficiency for PRIS aided multi-user networks was highlighted by taking the realistic power consumption into account. Furthermore, 
the effect of phase shifting error on outage behaviors was evaluated for PRIS based wireless networks \cite{Tianxiong2021PRISOutage}. 
From the practical viewpoint, the authors in \cite{Alexandros2021PRISHardware} researched the erogdic throughput of PRIS with imperfect hardware conditions. The deployment issue of PRIS was discussed in \cite{Chen2022PRISperformance}, where the reflecting probability of PRIS enabled networks was investigated by exploiting stochastic geometry.
To shed light on the overlay mode, the authors of \cite{Yiyang2021PRISD2DNakagami} applied PRIS to device-to-device communications under Nakagami-$m$ fading channels. With economizing the energy consumption \cite{Bouanani2021PRISWPT}, the average error probability of PRIS based wireless power transfer (WPT) was surveyed comprehensively.
	
Obviously, combining PRIS and NOMA can effectively realize the enhanced spectrum efficiency and dynamic environment configuration for 6G networks \cite{Zhiguo2022RISNOMA,Basar2022PRISNOMA,Yuanwei2022PRISNOMA}.
In \cite{Arthur2020PRISNOMA}, the attractive advantages of applying PRIS to NOMA networks were highlighted from the viewpoint of user fairness and system scalability. As further progress, the authors in \cite{Beixiong2020PRISNOMA} revealed the superiorities of PRIS-NOMA relative to PRIS-orthogonal multiple access (OMA) from the viewpoint of different rate configurations. By using on-off control strategy \cite{Ding2020PIRS}, the effect of hardware impairments (HIS) on PRIS-NOMA networks performance was evaluated carefully.
The authors of \cite{Hou2020PRISNOMA} analyzed the PRIS-NOMA networks' user prioritization outage characteristics. Furthermore, the authors of \cite{Yue2022IRSNOMA} investigated ergodic data rate and energy efficiency of PRIS-NOMA by using suboptimal scheme. Conditioned on imperfect detection conditions, the couple error probability of large PRIS-NOMA networks was outlined in \cite{Bariah2021LargePRISNOMAperformance}.
The adoption of PRIS to help two-way NOMA communications was highlighted in \cite{Yue2022PRISTwoNOMA}, where the message was interchanged between user nodes. In addition, the authors of \cite{Hemanth2020PRISNOMAHardware} discussed the effect of imperfect HIS on outage performance for PRIS-NOMA. The ergodic capacity of PRIS-NOMA was elucidated by taking into consideration the HIS \cite{Shaikh2022PRISNOMAHardware}.
	
Except that the RIS works in the passive mode, it can also operate in the active mode by integrating an additional low-power amplifier on each reflecting element\cite{Amato2018ReflectionAmplifiers,Loncar2020ActiveMetasurfaces}. The type of active RIS (ARIS) has aroused widespread attention from both academic and industrial circles \cite{CuiTiejun2022ARIS,Linglong2022ActiveRIS,Basar2022AIRSEmpowered}.
ARIS is able to increase system capacity by controlling the phase and amplitude of reflected signal.
More particularly, the authors in \cite{CuiTiejun2022ARIS} presented a novel amplifier based metasurface to accomplish the frequency and space domain control of incident electromagnetic waves. In \cite{Linglong2022ActiveRIS}, the multiplicative fading influence caused by PRIS was relieved as a sacrifice of power dissipation on the active reflection element. Even though ARIS integrates a low-power amplifier at a slightly higher cost than PRIS, it is still less expensive than other repeaters (e.g., amplify-and-forward (AF) relay and decode-and-forward (DF) relay).
With the emphasis on this issue, the author of \cite{Basar2022AIRSEmpowered} proved that ARIS is capable of turning the double-fading path loss of PRIS into the additive form.
As a further step, the sub-connected structure of ARIS was presented to regulate each element's phase shift independently \cite{Linglong2022ActiveFully}.
The exploitation of ARIS to wireless communications was outlined in \cite{Liang2021ActiveRIS}, where the ARIS-assisted links were strengthened with the increase of active elements. The bit error probability of ARIS based wireless systems was evaluated thoroughly \cite{George2022ARIS}.
On closer inspection, the authors of \cite{You2021ActiveRIS} identified that ARIS should be arranged around users with the small amplification power.
In \cite{Pan2022ARIS}, the achievable rate of ARIS exceeded that of PRIS on the condition of same power budget. In contrast to PRIS \cite{Qingqing2021activeRIS}, the ARIS-WPT systems had the potential to actualize the enhanced energy-efficiency. On this basis, the authors in \cite{Chen2022ARISAccess} further discussed the system throughput of ARIS aided uplink NOMA communication.
	
\subsection{Motivation and Contributions}
The previous research treatises have supplied a comprehensive viewpoint that applying PRIS in NOMA networks can further enhance spectrum efficiency and cell coverage.As previously stated in \cite{Linglong2022ActiveRIS,Basar2022AIRSEmpowered}, the multiplicative fading influence is inherent in PRIS-assisted wireless communications, where the system performance is affected by the double path loss at longer distances.
To solve this problem, it may be necessary to install a large number of passive reflective elements or to place the PRIS close to the transmitter or receiver. However this approach is not feasible in practice.The appearance of ARIS's studies focuses on solving the problem of double path loss brought by PRIS. Although previous work has validated the advantages of applying PRIS to NOMA \cite{CoMP-NOMA}, treatises on the potential performance gains by integrating ARIS with NOMA are still in their infancy. In \cite{Qingqing2021activeRIS}, the throughput maximization of ARIS-NOMA was discussed under the energy-limited communication scenarios, while there is no detailed survey of the outage probability and ergodic data rate of network performance.
In addition, most of the known studies have been analyzed under ideal hardware conditions, whereas HIS usually have a non-negligible impact in practical applications \cite{Emil2013Hardware,Chongwen2022hardware}.
Hence it is also worth considering how HIS affects ARIS-NOMA networks' performance. In practice, the successive interference cancellation (SIC) process can lead to decoding errors due to factors such as instrumentation errors and incorrect propagation.
Sparked by the above discussions, we concentrate on how ARIS-NOMA networks with HIS perform over cascade Nakagami-$m$ fading channels where imperfect successive interference cancellation (ipSIC) results in residual interference. More precisely, new exact and approximative expressions of outage probability and ergodic data rate with ipSIC/perfect SIC (pSIC) are deduced of ARIS-NOMA networks.
The influence of HIS on ARIS-NOMA networks performance is taken into consideration. The power consumption of ARIS-NOMA-HIS networks is further highlighted on the condition of delay-limited and delay-tolerant transmission schemes. Depending on these specific works, the foremost contributions of this paper are briefly summed up as follows:
\begin{enumerate}
\item We introduce the application of ARIS to NOMA networks in the presence of HIS. We derive the exact and asymptotic expressions of outage probability with ipSIC/pSIC for ARIS-NOMA-HIS networks over cascade Nakagami-$m$ fading channels. Utilizing theoretical analysis, we further acquire diversity orders of ARIS-NOMA-HIS networks, which are related to both channel ordering and the number of reflective components. We derive the outage probability for ARIS-OMA-HIS networks as a benchmark for comparison.
\item  We contrast ARIS-NOMA-HIS and PRIS-NOMA-HIS with that of conventional DF and AF relays in terms of the outage behaviors. We further demonstrate that ARIS-NOMA-HIS has a higher outage probability than traditional cooperative communication systems. We regulate the distances between BS and ARIS, where deploying ARIS on BS side results in better outage behaviors. By adding ARIS elements number, we find that the outage performance of ARIS first gets improved and then deteriorates.
\item  We derive the ergodic data rate of nearby and distant users with regard to ARIS-NOMA-HIS and PRIS-NOMA-HIS networks, respectively. We note that users' ergodic data rate converges to the throughput ceiling at high signal-to-noise ratio (SNR). As the reflection amplification factor increases, ARIS-NOMA-HIS has the ability to deliver an improved ergodic data rate. We derive the ergodic data rate of ARIS-OMA-HIS networks, and verify that ARIS-NOMA-HIS performs is superior to ARIS-OMA-HIS in terms of ergodic data rate.
\item  We assess the system throughput of ARIS-NOMA-HIS and PRIS-NOMA-HIS networks under delay-limited and delay-tolerant transmission schemes. In delay-limited schemes, we compare the system throughput between ARIS-NOMA-HIS and PRIS-NOMA-HIS networks. In delay-tolerant schemes, we compare the rates of ARIS-NOMA-HIS and PRIS-NOMA-HIS and found that both reached an upper bound at high SNRs. We also compare the energy efficiency of the two modes and find that ARIS-NOMA-HIS networks perform much better than the other comparative benchmarks, which have a high level of stability.
\end{enumerate}
	
\subsection{Organization and Notations}
The rest of this paper is structured as follows. Section \ref{Network Model} introduces the network model of ARIS-NOMA-HIS. The exact and approximate outage probabilities of ARIS-NOMA-HIS networks are highlighted prudentially in Section \ref{Outage Probability}. Section \ref{Ergodic Rate} evaluate the ergodic data rates for ARIS-NOMA-HIS networks, and then present the numerical analyses in Section \ref{Numerical Results}. Finally, Section \ref{Conclusion} summarizes this manuscript. The proofs of the maths are gathered in the appendix.
	
The key notations in this paper are presented as follows. The symbol ${\left(  \cdot  \right)^H}$ represents conjugate transpose operation. The  cumulative distribution function (CDF) and probability density function (PDF) of random variable $X$ are denoted by ${F_X}\left(  \cdot  \right)$ and ${f_X}\left(  \cdot  \right)$, respectively. $\mathbb{E}\{\cdot\}$ and $\mathbb{D}\{\cdot\}$ denote the expectation and variance operations, respectively. $diag\left(\cdot\right)$ denotes the diagonal matrix with element $one$.
	
\section{Network Model}\label{Network Model}
\begin{figure}[t!]
	\centering
	\includegraphics[width= 2.8in, height=1.8in]{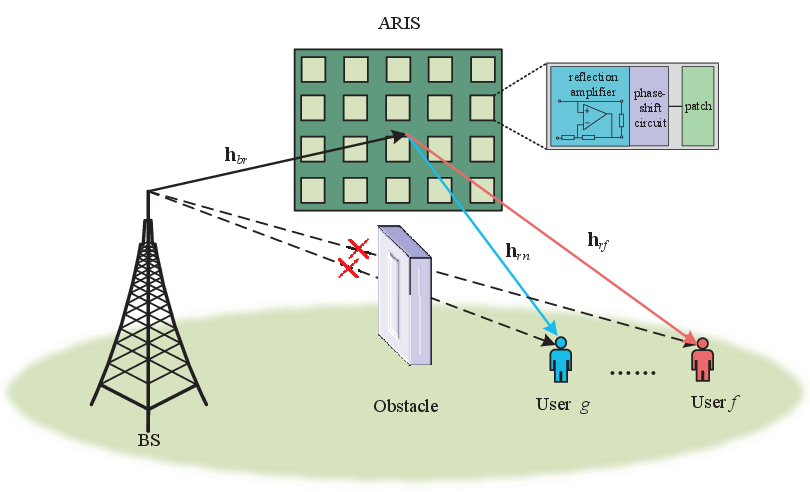}
	\caption{System model of ARIS-assisted NOMA-HIS communications, where ARIS is capable of amplifying the superposed signals, and then reflecting to users.}
	\label{System_model}
\end{figure}
We consider an ARIS-aided NOMA communication scenario with the presence of HIS, as shown in Fig. \ref{System_model}, in which the single antenna BS servers $K$ single antenna users by means of $L$ active reflecting elements. Specifically, existing ARIS architectures not only have the same circuitry to control phase shift as PRIS, but also integrate power amplifiers to amplify the radiated signal. The incident superposed signals at ARIS are amplified with a visible gain and then reflected toward the desired users. The inherent multiplicative attenuation in PRIS-assisted systems will be effectively tackled by active elements for the price of low power consumption.
In actuality, we focus attention on that BS-users' direct links obstructed by high-rise buildings \cite{WuTowards2020PRIS,Emil2020IRS}, where the deployment of ARIS can provide the line-of-sight across BS and non-orthogonal users. BS and users' side distortions are considered to more closely match practical scenarios.
To capture real channel conditions, we assume that the wireless links of ARIS-NOMA networks go through large-scale fading and Nakagami-$m$ fading\footnote{It is noteworthy that the derivations of performance analyses for ARIS-aided NOMA networks can be easily applied to multiple fading channels, i.e., the Gaussian, Reyleigh and Rician fading channels.}.
Let ${\bf{h}}_{br}={\sqrt \eta }d_{br}^{ - \frac{\alpha }{2}} \left[{h}_{br}^1,\cdots,{h}_{br}^l,\cdots,{h}_{br}^L\right]^{H}$ and ${\bf{h}}_{ri}=\sqrt \eta  d_{ri}^{ - \frac{\alpha }{2}} \left[{h}_{ri}^1,\cdots,{h}_{ri}^l,\cdots,{h}_{ri}^L\right]^{H}$ denote the baseband channel coefficients between BS and ARIS, and between ARIS and user $i$ ($1 \le i \le K$).
${h_{br}^l}$ and ${h_{ri}^l}$ denote the complex channel factors between BS and ARIS's $l$-th reflective element, and between ARIS's $l$-th reflective element and user $i$, respectively. $d_{br}$ and $d_{ri}$ are the distances of BS to ARIS and ARIS to user $i$, respectively. $\alpha$ is the path loss exponent, while $\eta  $ is the frequency dependent factor.
The reflection matrix at ARIS is defined as
	${\bf{\Phi}}={\sqrt{\beta}diag}\left({e}^{j{\theta_1}},\cdots,{e}^{j{\theta_l}},\cdots,{e}^{j{\theta_L}}\right)$, where ${\theta _l} \in \left[ {0,2\pi } \right)$ and ${\beta}$ indicate the $l$-th reflection element's phase shift and reflection amplification coefficient, respectively.
The reflection elements of ARIS are embedded in active amplifiers that have tunneling diodes or resistive converters. Based on the principle of electromagnetic scattering, each element can reflect and amplify an incoming radio frequency (RF) signal with an adjustable amplitude and phase. This allows ARIS to achieve reflection amplification coefficient greater than one, i.e., $\beta >1$ \cite{Linglong2022ActiveRIS,Liang2021ActiveRIS}. In practice, a negative resistive component such as a tunnel diode can be used to amplify the incident signal by converting the direct current bias power to RF power.
Broadly speaking, the baseband cascade channel gains transmitted by BS towards ARIS, and thereafter to the considerable users are sorted as
	${\left| {{\bf{h}}_{rK}^H{\bf{\Phi}} {{\bf{h}}_{br}}} \right|^2} >  \cdots  > {\left| {{\bf{h}}_{rg}^H{\bf{\Phi}} {{\bf{h}}_{br}}} \right|^2} >  \cdots  > {\left| {{\bf{h}}_{rf}^H{\bf{\Phi}} {{\bf{h}}_{br}}} \right|^2} >  \cdots  > {\left| {{\bf{h}}_{r1}^H{\bf{\Phi}} {{\bf{h}}_{br}}} \right|^2}$.
	Note that the positions of BS and ARIS are usually fixed in real scenarios, the effect of large-scale fading is essentially constant. Furthermore, the signal undergoes the same attenuation and phase shift modulation and is then reflected back to users by ARIS. This means that the difference in the signal-to-noise ratio (SINR) received by users is mainly caused by the large-scale fading in ${\bf{h}}_{ri}$, which are independent of each other \cite{Yue2020Unifid}. In addition, the hardware characteristics of receivers in a real situation can also have an impact on SINR, and these hardware characteristics can also be considered as independent variables. Therefore, in the following steps, the marginal distribution can be utilized to further calculate user's outage probability and ergodic data rata.
In particular, a couple of users, i.e., the near-end user $g$ and far-end user $f$ are picked out to carry out the NOMA mechanism. We suppose that the instantaneous channel state information (CSI) is accessible on BS and ARIS with help of channel evaluation or compressive sensing algorithms \cite{Zhengbeixiong2022IRSestimation,Swindlehurst2022IRSestimation}. To release this ideal assumption, the impact of imperfect CSI on system performance for ARIS-NOMA networks will be discussed in simulation results part.
	
\subsection{Signal Model}
The BS transfers superimposed signals, i.e., $x = \sqrt {{a_g}{P_b}} {x_g} + \sqrt {{a_f}{P_b}} {x_f}$ to user $\phi$ via ARIS, where $x_g$ and $x_f$ are normalized the unity power signals for the user $g$ and the user $f$, respectively. To warrant non-orthogonal users' fairness, the power allocation factors $a_g$ and $a_f$ for user $g$ and user $f$ satisfy the relations ${a_g}<{a_f}$ and ${a_g}+{a_f}=1$. $P_b$ denotes the BS transmission power. Unlike PRIS, thanks to the use of active elements like tunnel or Gunn diodes, ARIS will amplify the superimposed signals and the noise received at itself.
	Within these contexts, the received signals at user $\phi$ ($\phi  \in \left\{ {g, f} \right\}$) can be expressed as
\begin{align}\label{The received signals of users}
	{y_\phi } = {\sqrt \beta}{\bf{h}}_{r\phi}^H{\bf{\Theta}} {{\bf{h}}_{br}}\left( {x + {{\varepsilon }_b}} \right) + {\sqrt \beta}{\bf{h}}_{r\phi }^H{\bf{\Theta}} {{\bf{n}}_r} + {{\varepsilon }_\phi } + {n_\phi },
\end{align}
where ${\bf{\Theta}}={ diag}\left({e}^{j{\theta_1}},\cdots,{e}^{j{\theta_l}},\cdots,{e}^{j{\theta_L}}\right)$ denotes the phase shift matrix at ARIS.
	${{\varepsilon }_b}$ and ${{\varepsilon }_\phi}$ indicate the distortion from HIS at BS and user $\phi$ with ${{\varepsilon }_b}\sim{\cal C}{\cal N}\left({0,\kappa_b^2{P_b}}\right)$ and ${{\varepsilon }_\phi}\sim{\cal C}{\cal N}\left({0,\kappa_\phi^2{P_b}{{\left|{{\bf{h}}_{br}^H{\bf{\Theta}} {{\bf{h}}_{r\phi }}}\right|}^2}}\right)$ \cite{Emil2013Hardware}. The factors $\kappa_b$ and $\kappa_\phi$ indicate the level of HIS, which can be quantified by the size of the error vector.
	${\bf{n}}_r \sim {\cal C}{\cal N}\left( {0,{N_{tn}}{{\bf{I}}_L}} \right)$
	denotes the thermal noise at ARIS with the noise power $N_{tn}$, and  ${{\bf{I}}_L} \in \mathbb{C}{^{L \times 1}} $ is the identity matrix.
	$n_\phi$ indicates the Gaussian white noise generated at the user $\phi$ with ${n_\phi}\sim{\cal C}{\cal N}\left({0,{\sigma ^2}}\right)$.
	
Based on NOMA principle, user $f$'s information will be first detected by user $g$ with better channel conditions, then subtract it before decoding its own signal. Thus the corresponding received SINR can be written as
\begin{align}\label{The SINR of user g to detect user f}
	{\gamma _{g \to f}} = \frac{{\beta {P_b}{{\left| {{\bf{h}}_{rg}^H{\bf{\Theta }} {{\bf{h}}_{br}}} \right|}^2}{a_f}}}{{\beta {P_b}{{\left| {{\bf{h}}_{rg}^H{\bf{\Theta }} {{\bf{h}}_{br}}} \right|}^2}\chi  + \xi \beta {N_{tn}}{{\left\| {{\bf{h}}_{rg}^H{\bf{\Theta }} } \right\|}^2} + {\sigma ^2}}},
\end{align}
where ${\chi }={{a_g} + \kappa _b^2 + \kappa _g^2}$ and $\xi$ is the convert parameter. In ARIS networks $\xi$ is set to $1$. More precisely, when $\xi$ is equal to $zero$, ARIS will be divided into PRIS. After weeding out the information of user $f$, the receiving SINR of user $g$ decoding its own signal can be given by
\begin{align}\label{The SINR of user g}
	{\gamma _g} = \frac{{\beta {P_b}{{\left| {{\bf{h}}_{rg}^H{\bf{\Theta }} {{\bf{h}}_{br}}} \right|}^2}{a_g}}}{{\beta {P_b}{{\left| {{\bf{h}}_{rg}^H{\bf{\Theta }} {{\bf{h}}_{br}}} \right|}^2}{\chi _g} + \varpi {P_b}{{\left| {{h_{\rm{I}}}} \right|}^2} + \xi \beta {N_{tn}}{{\left\| {{\bf{h}}_{rg}^H{\bf{\Theta }} } \right\|}^2} + {\sigma ^2}}},
\end{align}
where ${{\chi _g}}={\kappa _b^2+\kappa_g^2}$ and $\varpi=1$ indicates the ipSIC case. ${h_I} \sim {\cal C}{\cal N}\left( {0,{\Omega _I}} \right)$ donotes the residual interference from ipSIC. Since the channel conditions of user $f$ are poor, it treats $x_g$ as interference and thus directly detects its own signal $x_f$. For the time being, the SINR for user $f$ detects $x_f$ can be written as
\begin{align}\label{The SINR of user f}
	{\gamma _f} = \frac{{\beta {P_b}{{\left| {{\bf{h}}_{rf}^H{\bf{\Theta }} {{\bf{h}}_{br}}} \right|}^2}{a_f}}}{{\beta {P_b}{{\left| {{\bf{h}}_{rf}^H{\bf{\Theta }} {{\bf{h}}_{br}}} \right|}^2}{\chi _f} + \xi \beta {N_{tn}}{{\left\| {{\bf{h}}_{rf}^H{\bf{\Theta }} } \right\|}^2} + {\sigma ^2}}},
\end{align}
where ${{\chi _f}} = \left( {{a_g} + \kappa _b^2 + \kappa _f^2} \right)$.
	
\subsection{ARIS-OMA with Hardware Impairments}
With the aim of display the performance advantages of ARIS-NOMA, we utilize the ARIS-OMA with HIS as a benchmark for comparison. Based on the previous subsection the SINR of user $o$ is given by
\begin{align}\label{The SINR of user o}
	{\gamma _o} = \frac{{\beta {P_b}{{\left| {{\bf{h}}_{ro}^H{\bf{\Theta }}{{\bf{h}}_{br}}} \right|}^2}}}{{\beta {P_b}{{\left| {{\bf{h}}_{ro}^H{\bf{\Theta }}{{\bf{h}}_{br}}} \right|}^2}{\chi _o} + \xi \beta {N_{tn}}{{\left\| {{\bf{h}}_{ro}^H{\bf{\Theta }}} \right\|}^2} + {\sigma ^2}}},
\end{align}
where ${{\bf{h}}_{ro}} = {\sqrt \eta }d_{ro}^{ - \frac{\alpha }{2}} \left[ {h_{ro}^1, \cdots ,h_{ro}^l, \cdots ,h_{ro}^L} \right]^H$ denotes the complex channel coefficient between ARIS and user $o$. $d_{ro}$ is the distance between ARIS and user $o$ and ${{\chi _o}} = \left( {\kappa _b^2 + \kappa _o^2} \right)$.
	
\section{Outage probability}\label{Outage Probability}
In this section, we evaluate ARIS-NOMA/OMA-HIS networks performance based on outage behaviour, where the approximate outage probability formulas for user $g$ and user $f$ are derived over cascade Nakagami-$m$ fading channels. Additionally, we approximate the thermal noise caused by ARIS as a constant for the purpose of facilitating computational analysis \cite{Hassibi2023,Jose2009}.
	
\subsection{The Outage Probability of User g}
By the sequence of the user decoding order in SIC progress, the outage incidents of user $g$ are defined as that: 1) User $g$ cannot detect user $f$'s signal $x_f$; 2) User $g$ can decode $x_f$, but cannot decode its own signal $x_g$. Its supplementary event can be easily represented that user $g$ is able to detect the signal $x_f$, and then decode the information $x_g$ \cite{Ding2014performance,Yue2018ExploitingNOMA}.
Under these circumstances, the outage probability of user $g$ with ipSIC for ARIS-NOMA-HIS networks is given by
	\begin{small}
		\begin{align}\label{the outage probability of user g}
			P_{ARIS,g}^{HIS} = {\rm{Pr}}\left( {{\gamma _{g \to f}} < {\gamma _{t{h_f}}}} \right) + {\rm{Pr}}\left( {{\gamma _{g \to f}} > {\gamma _{t{h_f}}},{\gamma _g} < {\gamma _{t{h_g}}}} \right),
		\end{align}
	\end{small}
	where $\gamma_{th_{g}}={2^{R_{g}}-1}$ and $\gamma_{th_{f}}={2^{R_{f}}-1}$ denote the target SNRs of signals $x_g$ and $x_f$ detected by user $g$ and user $f$, respectively. $R_{g}$ and $R_{f}$ are the target rates of user ${g}$ and user $f$, respectively. In this case, the approximate outage probability for user $g$ is able to be given in the following theorem.
	
\begin{theorem}\label{Theorem1:the OP of user g}
	In terms of cascade Nakagami-$m$ fading channels, the outage probability of user $g$ with ipSIC for ARIS-NOMA-HIS networks can be approximated as
	\begin{align}\label{the OP with ipSIC of user g ARIS-NOMA-HIS}
		P_{{ARIS},g}^{{HIS,ipSIC}} \approx& {\Psi _g}\sum\limits_{k = 0}^{K - g} {{
				{K - g}  \choose
				k  }\frac{{{{\left( { - 1} \right)}^k}}}{{\left( {g + k} \right){\psi _g}}}}\sum\limits_{u = 1}^U {{H_u}} x_u^{{m_g} - 1}  \nonumber  \\
		&\times{\left[ {\gamma \left( {{b_g} + 1,\frac{{\sqrt {{\varphi _g}{x_u} + {\vartheta _g}} }}{{{c_g}\sqrt {\beta {m_g}} }}} \right)} \right]^{g + k}},
	\end{align}
	where ${\Psi _g} = \frac{{K!}}{{\left( {K - g} \right)!\left( {g - 1} \right)!}}$, ${\psi _g} = \Gamma \left( {m_g} \right){\left[ {\Gamma \left( {{b_g} + 1} \right)} \right]^{g + k}}$, ${\varphi _g} = d_{br}^\alpha d_{rg}^\alpha {\varsigma _g}\varpi {P_b} $, ${\vartheta _g} = {m_g}d_{br}^\alpha d_{rg}^\alpha {\varsigma _g}\left( {\xi \beta {N_{tn}}L{\omega _{rg}} + {\sigma ^2}} \right)$, ${\varsigma _g} = \frac{{{\gamma _{t{h_g}}}}}{{{P_b}\left( {{a_g} - {\gamma _{t{h_g}}}\chi _g^2} \right)}}$ with the condition of ${a_g} > {\gamma _{t{h_g}}}{\chi _g^2}$,
	${b_g} = \frac{{L\mu _g^2}}{{{\Omega _g}}} - 1$, ${c_g} = \frac{{{\Omega _g}}}{{{\mu _g}}}$, ${\mu _g} = \frac{{\Gamma \left( {{m_r} + \frac{1}{2}} \right)\Gamma \left( {{m_g} + \frac{1}{2}} \right)}}{{\Gamma \left( {{m_r}} \right)\Gamma \left( {{m_g}} \right)\sqrt {{m_r}{m_g}} }}$, ${\Omega _g} = 1 - \frac{1}{{{m_r}{m_g}}}{\left[ {\frac{{\Gamma \left( {{m_r} + \frac{1}{2}} \right)\Gamma \left( {{m_g} + \frac{1}{2}} \right)}}{{\Gamma \left( {{m_r}} \right)\Gamma \left( {{m_g}} \right)}}} \right]^2}$,
	$m_r$ and $m_g$ denote the Nakagami-$m$ fading parameters from the BS to ARIS, and from ARIS to non-orthogonal user $g$, respectively.
	$\Gamma \left(  \cdot  \right)$ is the gamma function \cite[Eq. (8.310.1)]{2000gradshteyn}.
	$H_u$ and $x_u$ represent the Gauss-Laguerre integration's weight and abscissas, separately.
	Particularly, $x_u$ is the $u$-th zero point of Laguerre polynomial ${L_U\left(x\right)}={\frac{e^x}{U!}}{\frac{d}{dx^U}}{\left(x^Ue^{-x}\right)}$ and the $u$-th weight can be denoted by ${H_u}={\frac{\left({U!}\right)^{2}x_u}{\left[L_{U+1}\left(x_u\right)\right]^{2}}}$. $U$ is a tradeoff parametric to certify the complexity-accuracy. $\gamma \left( {a,x} \right) = \int_0^x {{t^{a - 1}}{e^{ - t}}dt} $ denotes the lower incomplete gamma function \cite[Eq. (8.350.1)]{2000gradshteyn}.
\end{theorem}
\begin{proof}
	See Appendix~A.
\end{proof}
	
\begin{corollary}\label{Corollary:the OP with pSIC of user g ARIS-NOMA-HIS}
	When the parameter $\varpi $ is set to be zero, the closed-form outage probability of user $g$ with pSIC for ARIS-NOMA-HIS can be given by
	\begin{align}\label{the OP with pSIC of user g ARIS-NOMA-HIS}
		P_{ARIS,g}^{HIS,pSIC} =& {\Psi _g}\sum\limits_{k = 0}^{K - g} {{
				{K - g}  \choose
				k  }\frac{{{{\left( { - 1} \right)}^k}}}{{g + k}}} \left[ {\frac{1}{{\Gamma \left( {{b_g} + 1} \right)}}} \right. \nonumber \\
		&  \times {\left. {\gamma \left( {{b_g} + 1,\frac{{\sqrt {{\vartheta _{g,pSIC}}} }}{{{c_g}\sqrt \beta  }}} \right)} \right]^{g + k}},
	\end{align}
	where ${\vartheta _{g,pSIC}} = d_{br}^\alpha d_{rg}^\alpha {\varsigma _g}\left( {\xi \beta {N_{tn}}L{\omega _{rg}} + {\sigma ^2}} \right)$.
\end{corollary}
	
\begin{corollary}\label{Corollary:the OP with ipSIC of user g PRIS-NOMA-HIS}
	When the parameter $\xi $ is set to be zero, the outage probability of user $g$ with ipSIC for PRIS-NOMA-HIS is approximated as
	\begin{align}\label{the OP with ipSIC of user g PRIS-NOMA-HIS}
		P_{PRIS,g}^{HIS,ipSIC} \approx& {\Psi _g}\sum\limits_{k = 0}^{K - g} {{
				{K - g}  \choose
				k  }\frac{{{{\left( { - 1} \right)}^k}}}{{\left( {g + k} \right){\psi _g}}}}\sum\limits_{u = 1}^U {{H_u}} x_u^{{m_g} - 1} \nonumber \\
		&\times{\left[ {\gamma \left( {{b_g} + 1,\frac{{\sqrt {{\varphi _g}{x_u} + {\vartheta _{PRIS,g}}} }}{{{c_g}\sqrt {{m_g}} }}} \right)} \right]^{g + k}} ,
	\end{align}
	where ${\vartheta _{PRIS,g}} = {m_g}d_{br}^\alpha d_{rg}^\alpha {\varsigma _g}{\sigma ^2}$.
\end{corollary}
	
\begin{corollary}\label{Corollary:the OP with pSIC of user g PRIS-NOMA-HIS}
	When the parameter $\varpi  $ and $\xi $ are both set to be zero, the closed-form outage probability of user $g$ with pSIC for PRIS-NOMA-HIS can be written as
	\begin{align}\label{the OP with pSIC of user g PRIS-NOMA-HIS}
		P_{{PRIS},g}^{HIS,pSIC} =& {\Psi _g}\sum\limits_{k = 0}^{K - g} {{
				{K - g}  \choose
				k  }\frac{{{{\left( { - 1} \right)}^k}}}{{g + k}}} \left[ {\frac{1}{{\Gamma \left( {{b_g} + 1} \right)}}} \right.\nonumber \\
		&  \times {\left. {\gamma \left( {{b_g} + 1,\frac{{\sqrt {d_{br}^\alpha d_{rg}^\alpha {\varsigma _g}{\sigma ^2}} }}{{{c_g}}}} \right)} \right]^{g + k}}.
	\end{align}
\end{corollary}
	
\subsection{The Outage Probability of User f}
With respect to user $f$ with worse channel conditions, the interrupt is generated if it cannot detect or decode its own message $x_f$. In this case, the outage probability of user $f$ for ARIS-NOMA-HIS networks is given by
\begin{align}\label{the outage probability of user f}
	P_{ARIS,f}^{HIS} = {\rm{Pr}}\left( {{\gamma _f} < {\gamma _{t{h_f}}}} \right).
\end{align}
	
\begin{theorem}\label{Theorem2:the OP of user f}
	Conditioned on cascade Nakagami-$m$ fading channels, the closed-form outage probability of user $f$ for ARIS-NOMA-HIS networks can be expressed as
		\begin{align}\label{the OP of user f ARIS-NOMA-HIS}
			P_{ARIS,f}^{HIS} =&{\Psi _f}\sum\limits_{k = 0}^{K - f} {{
					{K - f}  \choose
					k  }\frac{{{{\left( { - 1} \right)}^k}}}{{f + k}}} \left[ {\frac{1}{{\Gamma \left( {{b_f} + 1} \right)}}} \right. \nonumber \\
			&  \times {\left. {\gamma \left( {{b_f} + 1,\frac{{\sqrt {{\varphi _f} + {\vartheta _f}} }}{{{c_f}\sqrt \beta  }}} \right)} \right]^{f + k}},
	\end{align}
	where ${\Psi _f} = \frac{{K!}}{{\left( {K - f} \right)!\left( {f - 1} \right)!}}$, ${\varphi _f} = \xi \beta {N_{tn}}d_{br}^\alpha d_{rf}^\alpha {\varsigma _f}L{\omega _{rf}}$, ${\vartheta _f} = d_{br}^\alpha d_{rf}^\alpha {\varsigma _f}{\sigma ^2}$, ${\varsigma _f} = \frac{{{\gamma _{t{h_f}}}}}{{{P_b}\left( {{a_f} - {\gamma _{t{h_f}}}{\chi _f}} \right)}}$ with the condition of ${a_f} > {\gamma _{t{h_f}}}{\chi _f}$, ${\psi _f} = \Gamma \left( {m_f} \right){\left[ {\Gamma \left( {{b_f} + 1} \right)} \right]^{f + k}}$, ${b_f} = \frac{{L\mu _f^2}}{{{\Omega _f}}} - 1$, ${c_f} = \frac{{{\Omega _f}}}{{{\mu _f}}}$, ${\mu _f} = \frac{{\Gamma \left( {{m_r} + \frac{1}{2}} \right)\Gamma \left( {{m_f} + \frac{1}{2}} \right)}}{{\Gamma \left( {{m_r}} \right)\Gamma \left( {{m_f}} \right)\sqrt {{m_r}{m_f}} }}$, ${\Omega _f} = 1 - \frac{1}{{{m_r}{m_f}}}{\left[ {\frac{{\Gamma \left( {{m_r} + \frac{1}{2}} \right)\Gamma \left( {{m_f} + \frac{1}{2}} \right)}}{{\Gamma \left( {{m_r}} \right)\Gamma \left( {{m_f}} \right)}}} \right]^2}$, $m_f$ denote the multipath fading parameter from ARIS to non-orthogonal user $f$.
\end{theorem}
\begin{proof}
	See Appendix~B.
\end{proof}
	
\begin{corollary}\label{Corollary:the OP of user f PRIS-NOMA-HIS}
	When the parameter $\xi$ is set to be zero, the closed-form outage probability of user $f$ for PRIS-NOMA-HIS can be given by
		\begin{align}\label{the OP of user f PRIS-NOMA-HIS}
			P_{PRIS,f}^{HIS} =& {\Psi _f}\sum\limits_{k = 0}^{K - f} {{
					{K - f}  \choose
					k  }\frac{{{{\left( { - 1} \right)}^k}}}{{f + k}}} \left[ {\frac{1}{{\Gamma \left( {{b_f} + 1} \right)}}} \right. \nonumber \\
			&  \times {\left. {\gamma \left( {{b_f} + 1,\frac{{\sqrt {d_{br}^\alpha d_{rf}^\alpha {\varsigma _f}{\sigma ^2}} }}{{{c_f}}}} \right)} \right]^{f + k}}.
	\end{align}
\end{corollary}
	
\subsection{The Outage Probability of User o}
Analogous to the above analytical progresses, the outage probability of orthogonal user $o$ for ARIS-OMA-HIS networks is given by
\begin{align}\label{the outage probability of user o}
	P_{ARIS,o}^{HIS}= {\rm{Pr}}\left( {{\gamma _o} < {\gamma _{t{h_o}}}} \right),
\end{align}
where ${\gamma _{t{h_o}}}$ denotes the target SNRs of user $o$.
\begin{theorem}\label{Theorem3:the OP of user o}
	Conditioned on cascade Nakagami-$m$ fading channels, the outage probability of user $o$ for ARIS-OMA-HIS networks can be expressed as
	\begin{align}\label{the OP of user o ARIS-OMA-HIS}
		P_{ARIS,o}^{HIS} = \frac{1}{{\Gamma \left( {{b_o} + 1} \right)}}\gamma \left( {{b_o} + 1,\frac{{\sqrt {{\varphi _o} + {\vartheta _o}} }}{{{c_o}\sqrt \beta  }}} \right),
	\end{align}
	where ${\varphi _o} = \xi \beta {N_{tn}}d_{br}^\alpha d_{ro}^\alpha {\varsigma _o}L{\omega _{ro}}$, ${\vartheta _o} = d_{br}^\alpha d_{ro}^\alpha {\varsigma _o}{\sigma ^2}$, ${\varsigma _o} = \frac{{{\gamma _{t{h_o}}}}}{{{P_b}\left( {1 - {\gamma _{t{h_o}}}{\chi _o}} \right)}}$ with the condition of ${\gamma _{t{h_o}}}{\chi _o} < 1$, and $m_o$ denote the multipath fading parameter from ARIS to orthogonal user $o$.
\end{theorem}
	
\begin{corollary}\label{Corollary:the OP of  user o PRIS-OMA-HIS}
	When the parameter $\xi$ is set to be zero, the closed-form outage probability of user $o$ for PRIS-OMA-HIS can be expressed as
	\begin{align}\label{the OP of user o PRIS-OMA-HIS}
		P_{PRIS,o}^{HIS} = \frac{1}{{\Gamma \left( {{b_o} + 1} \right)}}\gamma \left( {{b_o} + 1,\frac{{\sqrt {d_{br}^\alpha d_{ro}^\alpha {\varsigma _o}{\sigma ^2}} }}{{{c_o}}}} \right).
	\end{align}
\end{corollary}
	
\subsection{Diversity Analysis}\label{Diversity Analysis}
In this section, we pick the diversity order to analyze the outage behaviors for ARIS-NOMA-HIS networks. To put it in another way, the diversity order can describe the rate at which the probability of disruption decreases as SNR increases \cite{Yingjie2023RISNOMA}. As the grading order becomes boosted, the outage probability decays faster. To be more specific, the diversity order is defined as
\begin{align}\label{the definition of diversity order}
	d =  - \mathop {\lim }\limits_{\rho  \to \infty } \frac{{\log \left( {\Pr \left( \rho  \right)} \right)}}{{\log \left( \rho  \right)}},
\end{align}
where $\rho  = \frac{{{P_b}}}{{{\sigma ^2}}}$ means the transmit SNR. ${\Pr \left( \rho  \right)}$ represents the asymptotic outage probability in case of high SNRs.
	
\begin{corollary}\label{Corollary:the asymptotic OP of user g for ARIS}
	When $P_b$ approaches infinity, the asymptotic outage probability of user $g$ with ipSIC/pSIC for ARIS-NOMA-HIS can be respectively expressed as
	\begin{align}\label{the asymptotic OP with ipSIC of user g for ARIS}
		P_{{ARIS},g}^{{ipSIC,}\infty } \approx& {\Psi _g}\sum\limits_{k = 0}^{K - g} {{
				{K - g}  \choose
				k  }\frac{{{{\left( { - 1} \right)}^k}}}{{\left( {g + k} \right){\psi _g}}}} \nonumber\\
		& \times\sum\limits_{u = 1}^U {{H_u}} x_u^{{m_g} - 1} {\left[ {\gamma \left( {{b_g} + 1,\frac{{\sqrt {{\varphi _g}{x_u}} }}{{{c_g}\sqrt {\beta {m_g}} }}} \right)} \right]^{g + k}},
	\end{align}
	and
	\begin{align}\label{the asymptotic OP with pSIC of user g for ARIS}
		&P_{ARIS,g}^{pSIC,\infty } = {\Psi _g}\sum\limits_{k = 0}^{K - g} {{
				{K - g}  \choose
				k  }\frac{{{{\left( { - 1} \right)}^k}}}{{g + k}}} \nonumber\\
		&\times {\left\{ {{{\left[ {\frac{{\Gamma \left( {2{m_r}} \right)\Gamma \left( {2{m_g}} \right){\Upsilon _g}}}{{{\Lambda _g}\Gamma \left( {{m_r} + {m_g} + \frac{1}{2}} \right)}}} \right]}^L}\frac{{{{\left( {{\vartheta _{g,pSIC}}} \right)}^{L{m_r}}}{\beta ^{ - L{m_r}}}}}{{2L{m_r}\left( {2L{m_r} - 1} \right)!}}} \right\}^{g + k}},
	\end{align}
	where $\Upsilon_g  = {\Upsilon _1}F\left( {2{m_r},{m_r} - {m_g} + \frac{1}{2};{m_r} + {m_g} + \frac{1}{2};1} \right)$, ${\Upsilon _1} = {4^{{m_r} - {m_f} + 1}}\sqrt \pi  {\left( {{m_r}{m_f}} \right)^{{m_r}}}$ and $F\left( { \cdot , \cdot ; \cdot ; \cdot } \right)$ is the ordinary hypergeometric function \cite[Eq. (9.100)]{2000gradshteyn}. $\delta  = \left( {2L{m_r} - 1} \right)!$, ${\Lambda _g} = \Gamma \left( {{m_r}} \right)\Gamma \left( {{m_g}} \right)$.
\end{corollary}
	
\begin{corollary}\label{Corollary:the asymptotic OP of user g for ARIS}
	When $P_b$ approaches infinity, the asymptotic outage probability of user $g$ with ipSIC/pSIC for PRIS-NOMA-HIS can be respectively expressed as
	\begin{align}\label{the asymptotic OP with ipSIC of user g for PRIS}
		P_{PRIS,g}^{ipSIC,\infty } \approx& {\Psi _g}\sum\limits_{k = 0}^{K - g} {{
				{K - g}  \choose
				k  }\frac{{{{\left( { - 1} \right)}^k}}}{{\left( {g + k} \right){\psi _g}}}}  \nonumber\\
		&  \times \sum\limits_{u = 1}^U {{H_u}} x_u^{{m_g} - 1}{\left[ {\gamma \left( {{b_g} + 1,\frac{{\sqrt {{\varphi _g}{x_u}} }}{{{c_g}\sqrt {{m_g}} }}} \right)} \right]^{g + k}},
	\end{align}
	and
	\begin{align}\label{the asymptotic OP with pSIC of user g for PRIS}
		&P_{{PRIS},g}^{pSIC,\infty } = {\Psi _g}\sum\limits_{k = 0}^{K - g} {{
				{K - g}  \choose
				k  }\frac{{{{\left( { - 1} \right)}^k}}}{{g + k}}}  \nonumber\\
		& \times {\left\{ {{{\left[ {\frac{{\Gamma \left( {2{m_r}} \right)\Gamma \left( {2{m_g}} \right){\Upsilon _g}}}{{{\Lambda _g}\Gamma \left( {{m_r} + {m_g} + \frac{1}{2}} \right)}}} \right]}^L}\frac{{{{\left( {d_{br}^\alpha d_{rg}^\alpha {\varsigma _g}{\sigma ^2}} \right)}^{L{m_r}}}}}{{2L{m_r}\left( {2L{m_r} - 1} \right)!}}} \right\}^{g + k}}.
	\end{align}
\end{corollary}
	
\begin{remark}\label{Remark1:the diversity order of the user g}
	As $P_b$ approaches infinity, the outage probability can be expressed as a constant plus the power exponent of $P_b$. We note that diversity order is the power exponent, which determines the rate of $P_b$ converging to the error floor. Upon substituting \eqref{the asymptotic OP with ipSIC of user g for ARIS} and \eqref{the asymptotic OP with ipSIC of user g for PRIS} into \eqref{the definition of diversity order}, the diversity order of user $g$ with ipSIC is equivalent to $zero$. It  is due to the effect of residual interference from ipSIC.
	Upon substituting \eqref{the asymptotic OP with pSIC of user g for ARIS} and \eqref{the asymptotic OP with pSIC of user g for PRIS} into \eqref{the definition of diversity order}, the diversity order of user $g$ with pSIC for is equal to $Lm_r\left(g+k\right)$. We can find that diversity order is associated with the number of reflected elements and the order of channels and that the introduction of thermal noise from ARIS does not have an effect on the diversity order.
\end{remark}
	
\begin{corollary}\label{Corollary:the asymptotic OP of user f}
	When $P_b$ approaches infinity, the asymptotic outage probability of user $f$ for ARIS/PRIS-NOMA-HIS can be respectively given by
	\begin{align}\label{the asymptotic OP of user f for ARIS}
		&P_{ARIS,f}^{HIS,\infty } = {\Psi _f}\sum\limits_{k = 0}^{K - f} {{
				{K - f}  \choose
				k  }\frac{{{{\left( { - 1} \right)}^k}}}{{f + k}}} \nonumber\\
		& \times {\left\{ {{{\left[ {\frac{{\Gamma \left( {2{m_r}} \right)\Gamma \left( {2{m_f}} \right){\Upsilon _f}}}{{{\Lambda _f}\Gamma \left( {{m_r} + {m_f} + \frac{1}{2}} \right)}}} \right]}^L}\frac{{{{\left( {{\varphi _f} + {\vartheta _f}} \right)}^{L{m_r}}}}}{{2L{m_r}{\beta ^{L{m_r}}}\left( {2L{m_r} - 1} \right)!}}} \right\}^{f + k}},
	\end{align}
	and
	\begin{align}\label{the asymptotic OP of user f for PRIS}
		&P_{PRIS,f}^{HIS,\infty } = {\Psi _f}\sum\limits_{k = 0}^{K - f} {{
				{K - f}  \choose
				k  }\frac{{{{\left( { - 1} \right)}^k}}}{{f + k}}} \nonumber\\
		& \times {\left\{ {{{\left[ {\frac{{\Gamma \left( {2{m_r}} \right)\Gamma \left( {2{m_f}} \right){\Upsilon _f}}}{{{\Lambda _f}\Gamma \left( {{m_r} + {m_f} + \frac{1}{2}} \right)}}} \right]}^L}\frac{{{{\left( {d_{br}^\alpha d_{rf}^\alpha {\varsigma _f}{\sigma ^2}} \right)}^{L{m_r}}}}}{{2L{m_r}\left( {2L{m_r} - 1} \right)}}} \right\}^{f + k}},
	\end{align}
	where ${\Lambda _f} = \Gamma \left( {{m_r}} \right)\Gamma \left( {{m_f}} \right)$, ${\Upsilon _2} = {4^{{m_r} - {m_f} + 1}}\sqrt \pi  {\left( {{m_r}{m_f}} \right)^{{m_r}}}$, $\Upsilon_f  = {\Upsilon _2}F\left( {2{m_r},{m_r} - {m_f} + \frac{1}{2};{m_r} + {m_f} + \frac{1}{2};1} \right)$.
\end{corollary}
	
\begin{remark}\label{Remark2:the diversity order of the user f}
	Upon substituting \eqref{the asymptotic OP of user f for ARIS} and \eqref{the asymptotic OP of user f for PRIS} into \eqref{the definition of diversity order}, the diversity orders of user $f$ for ARIS-NOMA-HIS and PRIS-NOMA-HIS networks are equal to $Lm_r\left(f+k\right)$.
\end{remark}
	
\begin{corollary}\label{Corollary:the asymptotic OP of user o}
	When $P_b$ approaches infinity, the asymptotic outage probability of user $o$ for ARIS-OMA-HIS and PRIS-OMA-HIS can be separately expressed as
	\begin{align}\label{the asymptotic OP of user o for ARIS}
		P_{ARIS,o}^{HIS,\infty } = {\left[ {\frac{{\Gamma \left( {2{m_r}} \right)\Gamma \left( {2{m_o}} \right){\Upsilon _o}}}{{{\Lambda _o}\Gamma \left( {{m_r} + {m_o} + \frac{1}{2}} \right)}}} \right]^L}\frac{{{{\left( {{\varphi _o} + {\vartheta _o}} \right)}^{L{m_r}}}}}{{2L{m_r}\delta {\beta ^{L{m_r}}}}},
	\end{align}
	and
	\begin{align}\label{the asymptotic OP of user o for PRIS}
		P_{PRIS,o}^{HIS,\infty } = \frac{1}{{\Gamma \left( {{b_o} + 1} \right)}}\gamma \left( {{b_o} + 1,\frac{{\sqrt {d_{br}^\alpha d_{ro}^\alpha {\varsigma _o}{\sigma ^2}} }}{{{c_o}}}} \right),
	\end{align}
	where $\Upsilon_o  = {\Upsilon _3}F\left( {2{m_r},{m_r} - {m_o} + \frac{1}{2};{m_r} + {m_o} + \frac{1}{2};1} \right)$, ${\Upsilon _3} = {4^{{m_r} - {m_o} + 1}}\sqrt \pi  {\left( {{m_r}{m_o}} \right)^{{m_r}}}$.
\end{corollary}
	
\begin{remark}\label{Remark3:the diversity order of the user o}
	Upon substituting \eqref{the asymptotic OP of user o for ARIS} and \eqref{the asymptotic OP of user o for PRIS} into \eqref{the definition of diversity order}, the diversity orders of user $o$ for ARIS-OMA-HIS and PRIS-OMA-HIS networks are equal to $Lm_r$.
\end{remark}
	
\subsection{Delay-Limited Transmission}
In delay-limited schemes, BS sends information at a steady rate and breaks based on the random fading of the wireless channel. The corresponding system throughput is given by
\begin{align}\label{the definition of throughput}
	{R_{ARIS ,dl}^{HIS,\Lambda }} = \left( {1 - P_{ARIS,g}^{HIS,\Lambda }} \right){R_g} + \left( {1 - P_{ARIS,f}^{HIS}} \right){R_f},
\end{align}
where $\Lambda  \in \left\{ {ipSIC,pSIC} \right\}$, ${P_{ARIS,g}^{HIS,ipSIC}}$, ${P_{ARIS,g}^{HIS,pSIC}}$ and ${P_{ARIS,f}^{HIS}}$ can be obtained from ${\left(\ref{the OP with ipSIC of user g ARIS-NOMA-HIS}\right)}$, ${\left(\ref{the OP with pSIC of user g ARIS-NOMA-HIS}\right)}$ and ${\left(\ref{the OP of user f ARIS-NOMA-HIS}\right)}$, respectively.
	
\section{Ergodic Data rate}\label{Ergodic Rate}
Ergodic data rate is also an important index for evaluating the performance of a communication system. When channel conditions determine the user's target rate, the ergodic data rate can be expressed as the maximum rate at which the system can transmit correctly over cascade Nakagami-$m$ fading channels and it can be defined as
\begin{align}\label{the definition of ergodic rate}
	R = \mathbb{E}\left[ {\log \left( {1 + \gamma_{\phi} } \right)} \right],
\end{align}
where $\gamma_{\phi}$ denotes the receiving SINR of user $\phi$.
	
\subsection{The Ergodic Data Rate of ARIS networks}
The average user $g$ achievable rate can be derived from ${\left(\ref{the definition of ergodic rate}\right)}$, which is defined as $R_{ARIS,g}^{erg} = \mathbb{E}\left[ {\log \left( {1 + {\gamma _g}} \right)} \right]$. In this case, the expression of user $g$ with ipSIC/pSIC for ARIS-NOMA-HIS networks can be assigned in the following theorem.
\begin{theorem}\label{Theorem4:the ER of user g}
	Conditioned on cascade Nakagami-$m$ fading channels, the ergodic data rate of user $g$ with ipSIC for ARIS-NOMA-HIS networks can be approximated as
	\begin{align}\label{the ergodic rate with ipSIC of user g in ARIS}
		&R_{ARIS,g}^{HIS,ipSIC} \approx \frac{{\pi {a_g}}}{{2N{\chi _g}\ln 2}}\sum\limits_{n = 1}^N {\frac{{2{\chi _g}\sqrt {1 - {x_n}^2} }}{{2{\chi _g} + {a_g}\left( {{x_n}{\rm{ + }}1} \right)}}} \left\{ {1 - \sum\limits_{u = 1}^U {} } \right.\nonumber\\
		& \times {H_u}x_u^{{m_g} - 1}{\Psi _g}\sum\limits_{k = 0}^{K - g} {{
				{K - g}  \choose
				k  }\frac{{{{\left( { - 1} \right)}^k}}}{{\left( {g + k} \right)\Gamma \left( {{m_g}} \right)}}} \left[ {\frac{1}{{\left[ {\Gamma \left( {{b_g} + 1} \right)} \right]}}} \right.\nonumber\\
		& \times \left. {{{\left. {\gamma \left( {{b_g} + 1,\frac{{\sqrt {\left( {{x_n}{\rm{ + }}1} \right)\left( {{\varpi _{g1}}{x_u} + {\varpi _{g2}}} \right)} }}{{{c_g}\sqrt {{m_g}{\tau _g}\left( {1 - {x_n}} \right)} }}} \right)} \right]}^{g + k}}} \right\},
	\end{align}
	where ${\tau _g} = \beta {P_b}{\chi _g}$, ${\varpi _{g1}} = d_{br}^\alpha d_{rg}^\alpha \varpi {P_b}$ and ${\varpi _{g2}} = d_{br}^\alpha d_{rg}^\alpha {m_g}\left( {\xi \beta {N_{tn}}L{\omega _{rg}} + {\sigma ^2}} \right)$. ${x_n} = \cos \left( {\frac{{2n - 1}}{{2N}}\pi } \right)$ is the abscissas in the Gauss-Chebyshev integral formula \cite[Eq. (25.4.38)]{MathematicalFunctions}.
\end{theorem}
\begin{proof}
	See Appendix~C.
\end{proof}
	
\begin{corollary}\label{Corollary:the ER with pSIC of ARIS for user g }
	When the parameter $\varpi $ is set to be zero, the ergodic data rate of user $g$ with pSIC for ARIS-NOMA-HIS over cascade Nakagami-$m$ fading channels can be approximated as
	\begin{align}\label{the ergodic rate with pSIC of user g in ARIS}
		&R_{ARIS,g}^{HIS,pSIC} \approx \frac{{\pi {a_g}}}{{2N{\chi _g}\ln 2}}\sum\limits_{n = 1}^N {\frac{{2{\chi _g}\sqrt {1 - {x_n}^2} }}{{2{\chi _g} + {a_g}\left( {{x_n}{\rm{ + }}1} \right)}}} \nonumber\\
		& \times \left\{ {1 - {\Psi _g}\sum\limits_{k = 0}^{K - g} {{
					{K - g}  \choose
					k  }\frac{{{{\left( { - 1} \right)}^k}}}{{g + k}}} } \right.\left[ {\frac{1}{{\Gamma \left( {{b_g} + 1} \right)}}} \right.\nonumber\\
		&  \times \left. {{{\left. {\gamma \left( {{b_g} + 1,\frac{{\sqrt {{\varpi _{g2,pSIC}}\left( {{x_n}{\rm{ + }}1} \right)} }}{{{c_g}\sqrt {{\tau _g}\left( {1 - {x_n}} \right)} }}} \right)} \right]}^{g + k}}} \right\},
	\end{align}
	where ${{\varpi '}_{g2}} = d_{br}^\alpha d_{rg}^\alpha \left( {\xi \beta {N_{tn}}L{\omega _{rg}} + {\sigma ^2}} \right)$.
\end{corollary}
	
\begin{theorem}\label{Theorem5:the ER of user f}
	Conditioned on cascade Nakagami-$m$ fading channels, the ergodic data rate of user $f$ for ARIS-NOMA-HIS networks is approximated as
	\begin{align}\label{the ergodic rate of user f in ARIS}
		&R_{ARIS,f}^{HIS} \approx \frac{{\pi {a_f}}}{{2\ln 2N{\chi _f}}}\sum\limits_{n = 1}^N {\frac{{2{\chi _f}\sqrt {1 - {x_n}^2} }}{{2{\chi _f} + {a_f}\left( {{x_n}{\rm{ + }}1} \right)}}} \nonumber\\
		& \times \left\{ {1 - {\Psi _f}\sum\limits_{k = 0}^{K - f} {{
					{K - f}  \choose
					k  }\frac{{{{\left( { - 1} \right)}^k}}}{{f + k}}} } \right.\left[ {\frac{1}{{\Gamma \left( {{b_f} + 1} \right)}}} \right. \nonumber\\
		&  \times \left. {{{\left. {\gamma \left( {{b_f} + 1,\frac{{\sqrt {\left( {{x_n}{\rm{ + }}1} \right)\left( {{\varpi _{f1}} + {\varpi _{f2}}} \right)} }}{{{c_f}\sqrt {{\tau _f}\left( {1 - {x_n}} \right)} }}} \right)} \right]}^{f + k}}} \right\},
	\end{align}
	where ${\varpi _{f1}} = \xi \beta {N_{tn}}d_{br}^\alpha d_{rf}^\alpha L{\omega _{rf}}$, ${\varpi _{f2}} = d_{br}^\alpha d_{rf}^\alpha {\sigma ^2}$ and ${\tau _f} = \beta {P_b}{\chi _f}$.
\end{theorem}
\begin{proof}
	See Appendix~D.
\end{proof}
	
\begin{theorem}\label{Theorem6:the ER of user o}
	Conditioned on cascade Nakagami-$m$ fading channels, the ergodic data rate of the orthogonal user $o$ for ARIS-OMA-HIS networks is approximated as
	\begin{align}\label{the ergodic rate of user o in ARIS}
		&R_{ARIS,o}^{HIS} \approx \frac{\pi }{{2N{\chi _o}\ln 2}}\sum\limits_{n = 1}^N {\frac{{2{\chi _o}\sqrt {1 - {x_n}^2} }}{{2{\chi _o} + \left( {{x_n}{\rm{ + }}1} \right)}}} \nonumber\\
		& \times \left[ {1 - \frac{1}{{\Gamma \left( {{b_o} + 1} \right)}}\gamma \left( {{b_o} + 1,\frac{{\sqrt {\left( {{x_n}{\rm{ + }}1} \right)\left( {{\varpi _{o1}} + {\varpi _{o2}}} \right)} }}{{{c_o}\sqrt {{\tau _o}\left( {1 - {x_n}} \right)} }}} \right)} \right],
	\end{align}
	where ${\tau _o} = \beta {P_b}{\chi _o}$, ${\varpi _{o1}} = \xi \beta {N_{tn}}d_{br}^\alpha d_{ro}^\alpha L{\omega _{ro}}$ and ${\varpi _{o2}} = d_{br}^\alpha d_{ro}^\alpha {\sigma ^2}$.
\end{theorem}
	
\subsection{The Ergodic Data Rate of PRIS networks}
The ergodic data rate of PRIS-NOMA-HIS networks can be obtained with the special case, i.e., $\beta = 1$, since the reflective element for the PRIS does not carry amplification.
\begin{corollary}\label{Corollary:the ER with ipSIC of PRIS for user g }
	In the unique scenario of $\beta = 1$, the ergodic data rate of user $g$ with ipSIC for PRIS-NOMA-HIS over cascade Nakagami-$m$ fading channels is approximated as
	\begin{align}\label{the ergodic rate with ipSIC of user g in PRIS}
		&R_{PRIS,g}^{HIS,ipSIC} \approx \frac{{\pi {a_g}}}{{2N{\chi _g}\ln 2}}\sum\limits_{n = 1}^N {\frac{{2{\chi _g}\sqrt {1 - {x_n}^2} }}{{2{\chi _g} + {a_g}\left( {{x_n}{\rm{ + }}1} \right)}}} \left\{ {1 - \sum\limits_{u = 1}^U {} } \right.  \nonumber\\
		& \times {H_u}x_u^{{m_g} - 1}{\Psi _g}\sum\limits_{k = 0}^{K - g} {{
				{K - g}  \choose
				k  }\frac{{{{\left( { - 1} \right)}^k}}}{{\left( {g + k} \right)\Gamma \left( {{m_g}} \right)}}} \left[ {\frac{1}{{\left[ {\Gamma \left( {{b_g} + 1} \right)} \right]}}} \right. \nonumber\\
		& \times \left. {{{\left. {\gamma \left( {{b_g} + 1,\frac{{\sqrt {\left( {{x_n}{\rm{ + }}1} \right)\left( {{\varpi _{g1}}{x_u} + {{\varpi ''}_{g2}}} \right)} }}{{{c_g}\sqrt {{m_g}{\tau _g}\left( {1 - {x_n}} \right)} }}} \right)} \right]}^{g + k}}} \right\},
	\end{align}
	where ${{\varpi ''}_{g2}} = d_{br}^\alpha d_{rg}^\alpha {m_g}{\sigma ^2}$.
\end{corollary}
	
\begin{corollary}\label{Corollary:the ER with pSIC of PRIS for user g }
	On the condition of $\varpi =0$ and $\beta = 1$, the ergodic data rate with pSIC of user $g$ for PRIS-NOMA-HIS over cascade Nakagami-$m$ fading channels is approximated as
	\begin{align}\label{the ergodic rate with pSIC of user g in PRIS}
		&R_{PRIS,g}^{HIS,pSIC} \approx \frac{{\pi {a_g}}}{{2N{\chi _g}\ln 2}}\sum\limits_{n = 1}^N {\frac{{2{\chi _g}\sqrt {1 - {x_n}^2} }}{{2{\chi _g} + {a_g}\left( {{x_n}{\rm{ + }}1} \right)}}}  \nonumber\\
		&  \times \left\{ {1 - {\Psi _g}\sum\limits_{k = 0}^{K - g} {{
					{K - g}  \choose
					k  }\frac{{{{\left( { - 1} \right)}^k}}}{{g + k}}} } \right. \left[ {\frac{1}{{\Gamma \left( {{b_g} + 1} \right)}}} \right.\nonumber\\
		&  \times \left. {{{\left. {\gamma \left( {{b_f} + 1,\frac{{\sqrt {d_{br}^\alpha d_{rg}^\alpha {\sigma ^2}\left( {{x_n}{\rm{ + }}1} \right)} }}{{{c_g}\sqrt {{\tau _g}\left( {1 - {x_n}} \right)} }}} \right)} \right]}^{g + k}}} \right\}.
	\end{align}
\end{corollary}
	
\begin{corollary}\label{Corollary:the ER of PRIS for user f }
	On the condition of $\beta = 1$, the ergodic data rate of user $f$ for PRIS-NOMA-HIS over cascade Nakagami-$m$ fading channels is approximated as
	\begin{align}\label{the ergodic rate of user f in PRIS}
		&R_{PRIS,f}^{HIS} \approx \frac{{\pi {a_f}}}{{2\ln 2N{\chi _f}}}\sum\limits_{n = 1}^N {\frac{{2{\chi _f}\sqrt {1 - {x_n}^2} }}{{2{\chi _f} + {a_f}\left( {{x_n}{\rm{ + }}1} \right)}}} \nonumber\\
		& \times \left\{ {1 - {\Psi _f}\sum\limits_{k = 0}^{K - f} {{
					{K - f}  \choose
					k  }\frac{{{{\left( { - 1} \right)}^k}}}{{f + k}}} } \right.\left[ {\frac{1}{{\Gamma \left( {{b_f} + 1} \right)}}} \right.\nonumber\\
		& \times \left. {{{\left. {\gamma \left( {{b_f} + 1,\frac{{\sqrt {\left( {{x_n}{\rm{ + }}1} \right){\varpi _{f2}}} }}{{{c_f}\sqrt {{\tau _f}\left( {1 - {x_n}} \right)} }}} \right)} \right]}^{f + k}}} \right\}.
	\end{align}
\end{corollary}
	
\begin{corollary}\label{Corollary:the ER of PRIS for user o }
	The ergodic data rate of the orthogonal user $o$ for PRIS-OMA-HIS networks can be expressed as
	\begin{align}\label{the ergodic rate of user o in PRIS}
		&R_{PRIS,o}^{HIS} \approx \frac{\pi }{{2\ln 2N{\chi _o}}}\sum\limits_{n = 1}^N {\frac{{2{\chi _o}\sqrt {1 - {x_n}^2} }}{{2{\chi _o} + \left( {{x_n}{\rm{ + }}1} \right)}}} \nonumber\\
		& \times \left[ {1 - \frac{1}{{\Gamma \left( {{b_o} + 1} \right)}}\gamma \left( {{b_o} + 1,\frac{{\sqrt {d_{br}^\alpha d_{ro}^\alpha {\sigma ^2}\left( {{x_n}{\rm{ + }}1} \right)} }}{{{c_o}\sqrt {{\tau _o}\left( {1 - {x_n}} \right)} }}} \right)} \right].
	\end{align}
\end{corollary}
	
\subsection{Multiplexing Gains Analysis}
The multiplexing gains at high SNRs is used as an indicators evaluation metric aimed at capturing the diversity of traversal rates using the transmit SNRs \cite{Lizhong2003,Xiejin2023STAR}, and can be given by
\begin{align}\label{the definition of slope analysis}
	S =  - \mathop {\lim }\limits_{\rho  \to \infty } \frac{{R_{erg}^\infty \left( \rho  \right)}}{{\log \left( \rho  \right)}},
\end{align}
where ${R_{erg}^\infty \left( {{\rho}} \right)}$ is the asymptotic ergodic data rate at high SNRs.
	
Based on \eqref{the ergodic rate with ipSIC of user g in ARIS}, we find the ergodic data rate of user $g$ reaches a throughput ceiling when $P_b$ tends to infinity, which can be respectively approximated as
\begin{align}\label{the throughput ceiling with ipSIC of user g}
	&R_{ARIS,g}^{ipSIC,\infty } \approx \frac{{\pi {a_g}}}{{2N{\chi _g}\ln 2}}\sum\limits_{n = 1}^N {\frac{{2{\chi _g}\sqrt {1 - {x_n}^2} }}{{2{\chi _g} + {a_g}\left( {{x_n}{\rm{ + }}1} \right)}}} \nonumber\\
	& \times \left\{ {1 - \sum\limits_{u = 1}^U {{H_u}} x_u^{{m_g} - 1}{\Psi _g}\sum\limits_{k = 0}^{K - g} {{
				{K - g}  \choose
				k  }\frac{{{{\left( { - 1} \right)}^k}}}{{\left( {g + k} \right)\Gamma \left( {{m_g}} \right)}}} } \right.\nonumber\\
	& \times \left. {{{\left[ {\frac{1}{{\left[ {\Gamma \left( {{b_g} + 1} \right)} \right]}}\gamma \left( {{b_g} + 1,\frac{{\sqrt {d_{br}^\alpha d_{rg}^\alpha {x_u^{{m_g} - 1}}} }}{{{c_g}\sqrt {\beta {m_g}{\chi _g}\left( {1 - {x_n}} \right)} }}} \right)} \right]}^{g + k}}} \right\},
\end{align}
and
\begin{align}\label{the throughput ceiling with pSIC of user g}
	R_{ARIS,g}^{pSIC,\infty } \approx \frac{{\pi {a_g}}}{{2N{\chi _g}\ln 2}}\sum\limits_{n = 1}^N {\frac{{2{\chi _g}\sqrt {1 - {x_n}^2} }}{{2{\chi _g} + {a_g}\left( {{x_n}{\rm{ + }}1} \right)}}}.
\end{align}
	
\begin{remark}\label{Remark4:the throughput ceiling of the user g}
	Substituting ${\left(\ref{the throughput ceiling with ipSIC of user g}\right)}$ and ${\left(\ref{the throughput ceiling with pSIC of user g}\right)}$ into ${\left(\ref{the definition of slope analysis}\right)}$, the multiplexing gains for user $g$ for ARIS-NOMA-HIS are equal to $zero$, which is inconsistent with traditional NOMA.
\end{remark}
	
According to \eqref{the ergodic rate of user f in ARIS}, the throughput ceiling on the ergodic data rate of user $f$ when $P_b$ tends to infinity can be written by
\begin{align}\label{the throughput ceiling of user f}
	R_{ARIS,f}^{HIS,\infty }  = \frac{{\pi {a_f}}}{{2N{\chi _f}\ln 2}}\sum\limits_{n = 1}^N {\frac{{2{\chi _f}\sqrt {1 - {x_n}^2} }}{{2{\chi _f} + {a_f}\left( {{x_n}{\rm{ + }}1} \right)}}}.
\end{align}
	
\begin{remark}\label{Remark5:the throughput ceiling of the user f}
	Substituting ${\left(\ref{the throughput ceiling of user f}\right)}$ into ${\left(\ref{the definition of slope analysis}\right)}$, user $f$'s multiplexing gains for ARIS-NOMA-HIS is also equivalent to $zero$. This is because the HIS has an effect on the system.
\end{remark}
	
According to \eqref{the ergodic rate of user o in ARIS}, the throughput ceiling on the ergodic data rate of user $o$ when $P_b$ approaches infinity is
\begin{align}\label{the throughput ceiling of user o}
	R_{ARIS,o}^{HIS,\infty } = \frac{\pi }{{2N{\chi _o}\ln 2}}\sum\limits_{n = 1}^N {\frac{{2{\chi _o}\sqrt {1 - {x_n}^2} }}{{2{\chi _o} + \left( {{x_n}{\rm{ + }}1} \right)}}}.
\end{align}
	
\begin{remark}\label{Remark6:the throughput ceiling of the user o}
	Substituting ${\left(\ref{the throughput ceiling of user o}\right)}$ into ${\left(\ref{the definition of slope analysis}\right)}$, ARIS-OMA-HIS networks yield a $zero$ multiplexing gain of user $o$, which is caused by the the effect of HIS.
\end{remark}
	
\subsection{Delay-Tolerate Transmission}
In delay-tolerant schemes, BS transmits data at a steady rate depending on the user's channel condition constraints, using the ergodic capacity as an upper limit. Accordingly, the system throughput for ARIS-NOMA-HIS networks can be expressed as
\begin{align}\label{the definition of delay-tolerate}
	{R_{ARIS,dt}^{HIS,\Lambda}} = R_{ARIS,g}^{HIS,\Lambda } + R_{ARIS,f}^{HIS},
\end{align}
where $\Lambda  \in \left\{ {ipSIC,pSIC} \right\}$, $R_{ARIS,g}^{HIS,ipSIC}$, $R_{ARIS,g}^{HIS,pSIC}$ and $R_{ARIS,f}^{HIS}$ can be obtained from ${\left(\ref{the ergodic rate with ipSIC of user g in ARIS}\right)}$, ${\left(\ref{the ergodic rate with pSIC of user g in ARIS}\right)}$ and ${\left(\ref{the ergodic rate of user f in ARIS}\right)}$, respectively.
	
\section{Energy Efficiency}\label{Energy Efficiency}
For ARIS-NOMA-HIS networks, the system's total power consumption primarily comprises the power transmitted at BS, the power consumed at BS, the power consumed at ARIS, the output power of ARIS and the hardware power consumption at users. This can be denoted specifically by
\begin{align}\label{the total power consumption}
	{P_{total}} = \kappa {P_b} + {P_{BS}} + L{P_{RE}} + {P_{out}} + {P_{U,\phi }},
\end{align}
where $\kappa \buildrel \Delta \over = {\nu ^{ - 1}}$ with $\nu$ is the transmitter power amplifier efficiency, ${P_b}$ indicates the emitted power at BS. ${P_{BS}}$ indicates the total static power loss at BS, $L{P_{RE}}$ is the total hardware static power loss of $L$ reflector elements at ARIS.
${P_{out}}$ indicates the output power of ARIS and ${P_{U,\phi }}$ indicates the hardware static power consumption of the $\phi$-th user. It is possible to derive that the system's energy efficiency can be written as the ratio of sum rate to total power consumption of system, which is represented as
	\begin{align}\label{the energy efficiency}
		{\eta _{EE}} = \frac{{{R_\Phi }}}{{{P_{total}}}},
	\end{align}
	where ${R_\Phi } \in \left( {R_{ARIS,dl}^{HIS,\Lambda},R_{ARIS,dt}^{HIS,\Lambda}} \right)$. $R_{ARIS,dl}^{HIS,\Lambda}$ and $R_{ARIS,dt}^{HIS,\Lambda}$ can be obtained from ${\left(\ref{the definition of throughput}\right)}$ and ${\left(\ref{the definition of delay-tolerate}\right)}$, respectively.
	
\begin{table}[!t]
	\centering
	\caption{The simulation parameters for ARIS-NOMA.}
	\tabcolsep5pt
	\renewcommand\arraystretch{1.2} 
	\begin{tabular}{|l|l|}
		\hline
		Monte Carlo simulations repeated  &  ${10^6}$ iterations \\
		\hline
		\multirow{2}{*}{The power allocation factors of user $g$ and user $f$ }&  \multirow{1}{*}{$a_g=0.25$}   \\
		&  \multirow{1}{*}{$a_f=0.75$}   \\
		\hline
		\multirow{2}{*}{The targeted rates of user $g$ and user $f$ } & \multirow{1}{*}{$R_{{g}}=1.5$ BPCU}  \\
		& \multirow{1}{*}{$R_{{f}}=1.5$ BPCU} \\
		\hline
		\multirow{1}{*}{The distance between BS and ARIS  }
		& \multirow{1}{*}{$d_{br}=10$ m } \\
		\hline
		\multirow{1}{*}{The distance between ARIS and user $g$   }
		& \multirow{1}{*}{$d_{rg}=10$ m }  \\
		\hline
		\multirow{1}{*}{The distance between ARIS and user $f$   }
		& \multirow{1}{*}{$d_{rf}=20$ m }  \\
		\cline{1-2}
		\multirow{1}{*}{The distance between ARIS and user $o$   }
		& \multirow{1}{*}{$d_{ro}=30$ m }  \\
		\cline{1-2}
		The reflection coefficient of ARIS &  $\beta  = 5 $   \\
		\hline
		The thermal noise power at ARIS   &  $N_{tn} =-30  $ dBm   \\
		\hline
		The AWGN power at users  &  $\sigma ^2 =-20  $ dBm   \\
		\hline
		Pass loss expression   &  $\alpha  = 2.2 $   \\
			
		\cline{1-2}
	\end{tabular}
	\label{parameter}
\end{table}
\section{Simulation Results}\label{Numerical Results}
In this section, the numerical results are provided to verify the effectiveness of analytical results for ARIS-NOMA networks. We also consider the impact of HIS on outage probability, ergodic data rate and energy efficiency of ARIS-NOMA networks. The noise power and bandwidth are set to be ${\sigma ^2} =  - 174 + 10\log \left( {BW} \right)$ and $BW=1000$ MHz, respectively. For illustration purposes, the simulation results used are concluded in Table~\ref{parameter} \cite{Hou2020PRISNOMA,You2021ActiveRIS}, where the abridgement of BPCU denotes bit per channel use.
To facilitate comparison, PRIS-NOMA-HIS, PRIS-OMA-HIS are regarded as benchmarks for ARIS-NOMA-HIS.
In addition, conventional cooperative communication schemes, i,e., multi-antenna AF relaying is also selected to compare the behaviour of ARIS-NOMA/OMA-HIS.
	
Something deserving of clarification is that to demonstrate the accuracy of the results, we have kept the total power consumed by PRIS-NOMA-HIS and ARIS-NOMA-HIS the same. More specifically, $P_b^{ARIS} = P_{BS}^{ARIS} + P_{RIS}^{ARIS} + L\left( {{P_{SW}} + {P_{DC}}} \right)$, $P_b^{PRIS} = P_{BS}^{PRIS} + L{P_{SW}}$ and $P_b^{ARIS} = P_b^{PRIS} = {P_b}$, where $P_{BS}^{ARIS}$ and $P_{BS}^{PRIS}$ are the transmit power of BS in ARIS and PRIS networks, $P_{RIS}^{ARIS}$ is the signal output power in ARIS and $P_{RIS}^{ARIS} = P_{BS}^{ARIS}\beta {\left\| {{\bf{\Theta }}{{\bf{h}}_{br}}} \right\|^2} + \beta {N_{tn}}{\left\| {\bf{\Theta }} \right\|^2}$, ${P_{SW}}$ indicates the amount of power used by each RIS element's control circuit and phase shift switches, ${P_{DC}}$ means the direct current biasing power \cite{Pan2022ARIS}.
	
	
\begin{figure}[t!]
	\begin{center}
		\includegraphics[width=2.8in,  height=2.0in]{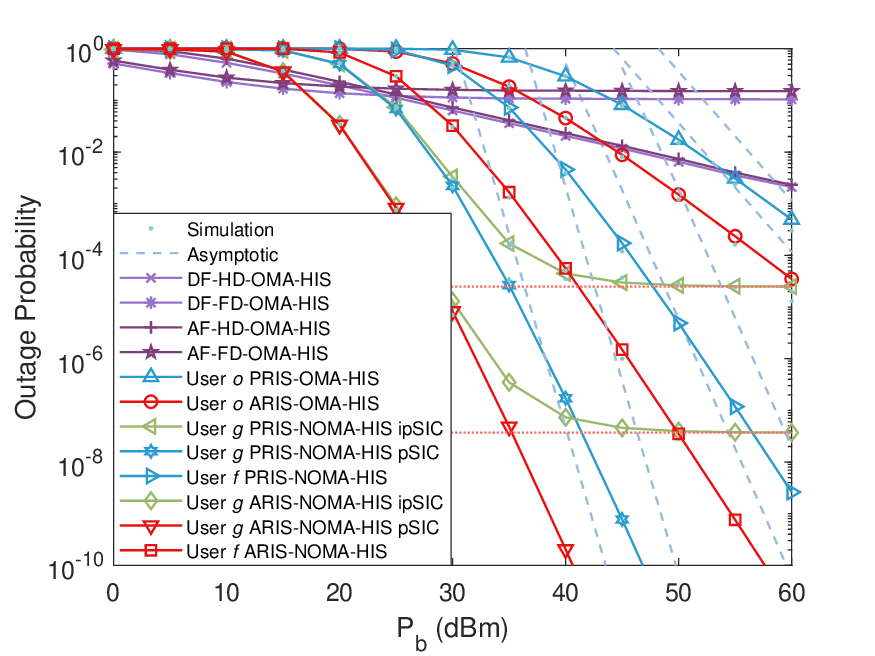}
		\caption{Outage probability versus the transmit power ${P_b}$, with $K=3$, $g=3$, $f=2$, $L=5$, $\beta=5$, and $m = 0.5$, $R_g=R_f=1.5$ BPCU.}
		\label{ARIS_NOMA_HI_add_Relay}
	\end{center}
\end{figure}
	
Fig. \ref{ARIS_NOMA_HI_add_Relay} plots the outage probability of ARIS-NOMA-HIS networks versus ${P_b}$ with $K=3$, $g=3$, $f=2$, $L=5$, $\beta=5$, $m = 0.5$ and $R_g=R_f=1.5$ BPCU.
	It shows that the  Monte Carlo simulation values of outage probability agree perfectly with the above theoretical analyses. We can find that the outage probability of user $g$ with ipSIC congregates to an error floor at higher transmit power $P_b$, which is identical to the conclusion in \textbf{Remark \ref{Remark1:the diversity order of the user g}}. ARIS-NOMA-HIS outage performance outperforms that of ARIS-OMA-HIS because NOMA has higher fairness than OMA when serving multiple customers at the same time.
	Another critical finding is that ARIS-OMA-HIS outperforms DF and AF relay switching between full-duplex (FD) and half-duplex (HD) modes for outages at high SNRs. The reason for it is that ARIS enhances the signals reflected back to users. This indicates that even though ARIS generates amplified noise, its outage performance is still superior to that of traditional cooperative communications, i.e., HD/FD AF relaying.
	
\begin{figure}[t!]
	\begin{center}
		\includegraphics[width=2.9in,  height=2.1in]{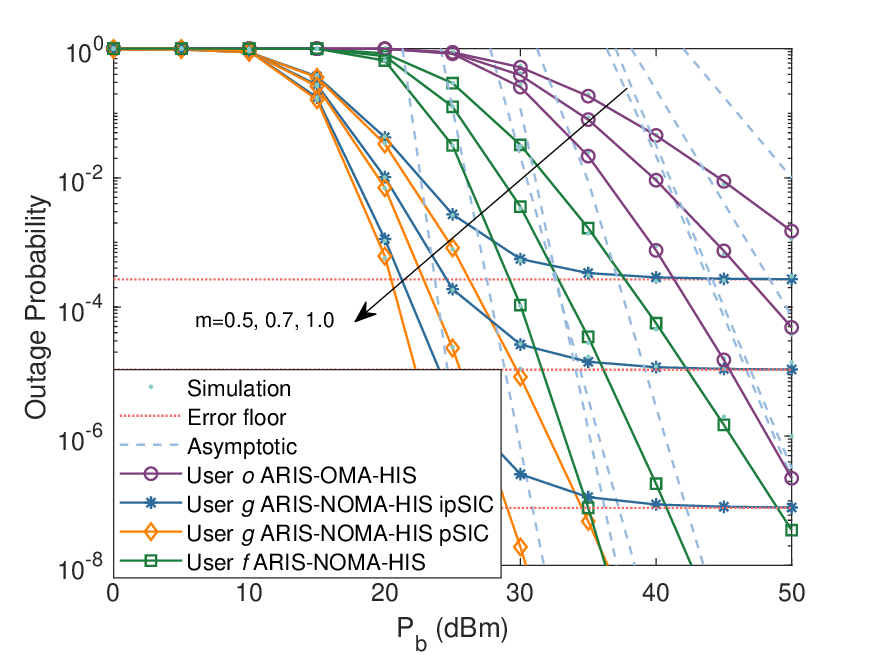}
		\caption{Outage probability versus the transmit power ${P_b}$, with $L=5$, $K=3$, $g=3$, $f=2$, $\beta=5$, and $R_g=R_f=1.5$ BPCU.}
		\label{ARIS_NOMA_HI_diff_m}
	\end{center}
\end{figure}
	
\begin{figure}[t!]
	\begin{center}
		\includegraphics[width=2.9in,  height=2.1in]{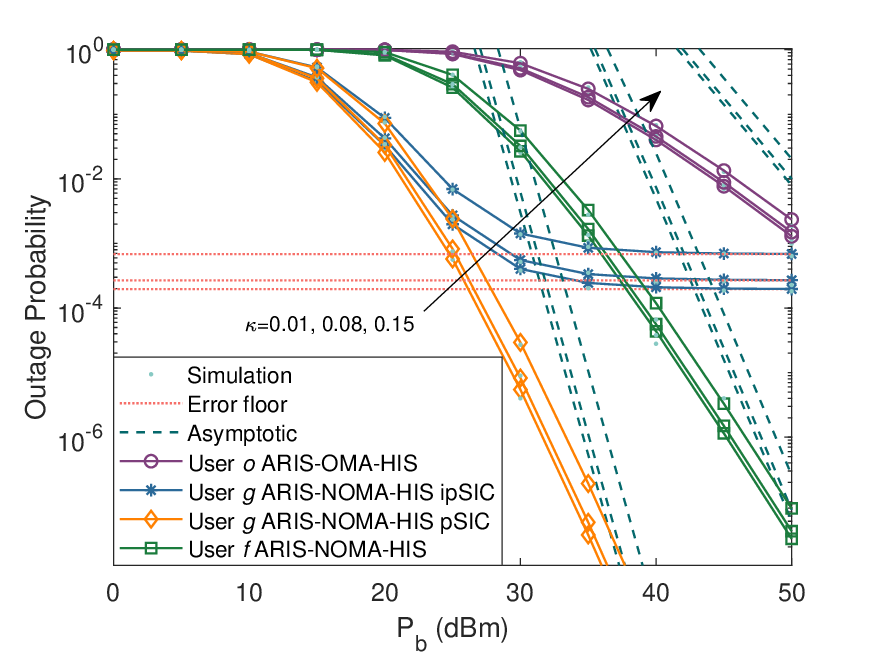}
		\caption{Outage probability versus the transmit power ${P_b}$, with $L=5$, $K=3$, $g=3$, $f=2$, $m=0.5$, $\beta=5$, and $R_g=R_f=1.5$ BPCU.}
		\label{ARIS_NOMA_HI_diff_kappa}
	\end{center}
\end{figure}
	
Fig. \ref{ARIS_NOMA_HI_diff_m} plots the outage probability for ARIS-NOMA/OMA-HIS networks versus ${P_b}$ for different values of the fading parameter $m$ with $K=3$, $L=5$, $g=3$, $f=2$, $\beta=5$ and $R_g=R_f=1.5$ BPCU.
It can be observed that the outage probability of ARIS-NOMA-HIS increases with the fading factor $m$. As $m$ becomes higher, i.e., $m=0.5$, $0.7$ and $1$, ARIS-NOMA-HIS networks are enabling improved outage performance.
The main reason for this phenomenon is that larger values of $m$ correspond to smaller channel fading, and when $m \to \infty$ indicates no fading. It is worth noting that when $m = 1$, the cascaded Nakagami-$m$ fading channel is converted to a Rayleigh fading channel.
As a further advance, Fig. \ref{ARIS_NOMA_HI_diff_kappa} shows the impact of HIS on system outage performance. It is visible that the outage behaviors of ARIS-NOMA-HIS are in connection with the HIS. As $\kappa$ increases, the gap between outage probabilities becomes larger at high SNRs.
It takes $P_b = 10$ to improve outage performance even if $\kappa =0.01$. The slope of the curve stays identical regardless of the value of $\kappa$ changes. This phenomenon indicates that we can improve the outage behaviours of ARIS-NOMA-HIS by reducing the HIS.
	
\begin{figure}[t!]
	\begin{center}
		\includegraphics[width=2.9in,  height=2.1in]{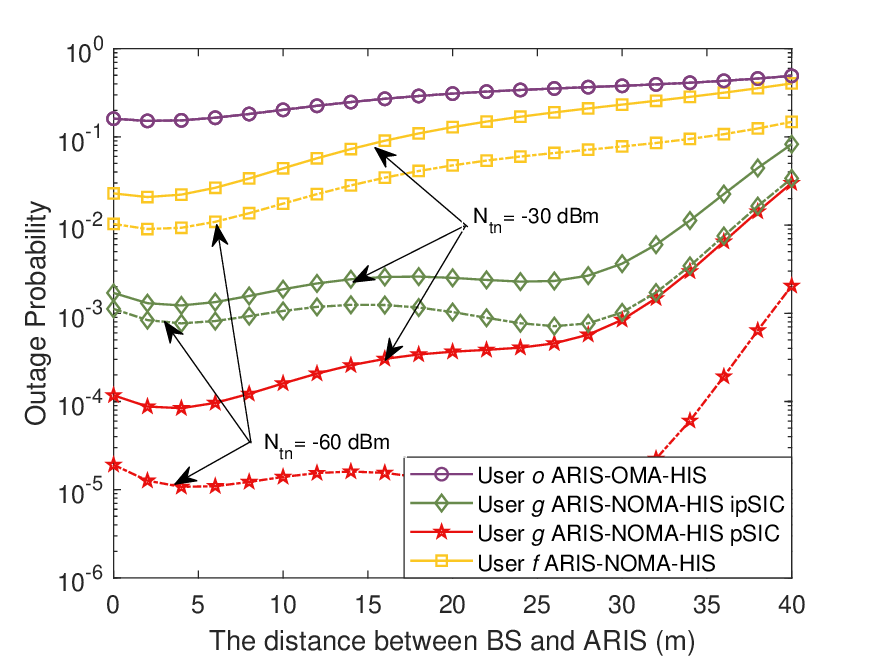}
		\caption{Outage probability versus the different distance between BS and ARIS, with $L=5$, $K=3$, $g=3$, $f=2$, $m=0.5$, $\beta=5$, and $R_g=R_f=1.5$ BPCU.}
		\label{ARIS_NOMA_HI_diff_dbr}
	\end{center}
\end{figure}
	
Fig. \ref{ARIS_NOMA_HI_diff_dbr} plots the outage probability of ARIS-NOMA/OMA-HIS networks versus the different distance from BS to ARIS. Assuming that the BS is in the same plane as user $g$ and user $f$, and they are located at $\left( {0,0} \right)$, $\left( {30,0} \right)$ and $\left( {40,-10} \right)$, respectively. Besides, ARIS are located at $\left( {x_{ARIS},10} \right)$, where ${x_{ARIS}}$ indicates the distance between BS and ARIS.
It can be noted that the outage performance of user $f$ is gradually deteriorating. The causes of this phenomenon is that as $d_{rf}$ decreases, the increasing trend of denominator in ${\left(\ref{The SINR of user f}\right)}$ is much greater than numerator. In summary, ARIS deployed on BS side offers higher performance gains. In addition, reducing the thermal noise values $N_{tn}$ does not improve the performance of ARIS-OMA-HIS networks.
	
\begin{figure}[t!]
	\begin{center}
		\includegraphics[width=2.9in,  height=2.1in]{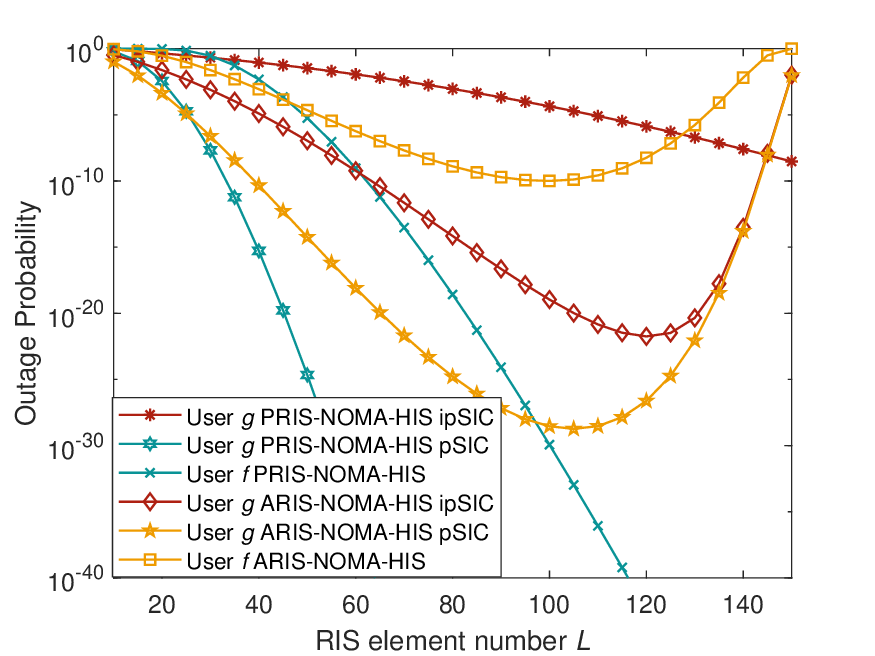}
		\caption{Outage probability versus the reflection elements number $L$, with $K=3$, $g=3$, $f=2$, $m=0.5$, $\beta=10$, and $R_g=R_f=1.5$ BPCU.}
		\label{ARIS_NOMA_HI_diff_L}
	\end{center}
\end{figure}
	
Fig. \ref{ARIS_NOMA_HI_diff_L} plots the outage probability of ARIS-NOMA-HIS networks versus the different reflection elements number $L$ with $P_b=20$ dBm, $K=3$, $g=3$, $f=2$, $m=0.5$, $\beta=10$ and $R_g=R_f=1.5$ BPCU. One can observe that initially both ARIS-NOMA-HIS and PRIS-NOMA-HIS achieve better performance when $L$ rises. However, as $L$ increase to a certain level, the outage performance of ARIS-NOMA-HIS starts to deteriorate. The reason for this phenomenon is that as $L$ increases, thermal noise from ARIS becomes larger.
This means that ARIS-NOMA-HIS networks performs about similar as PRIS-NOMA-HIS networks when there are fewer reflective elements. It is worth considering the trade-off between reflection elements number of ARIS and outage performance.
	
\begin{figure}[t!]
	\begin{center}
		\includegraphics[width=2.9in,  height=2.1in]{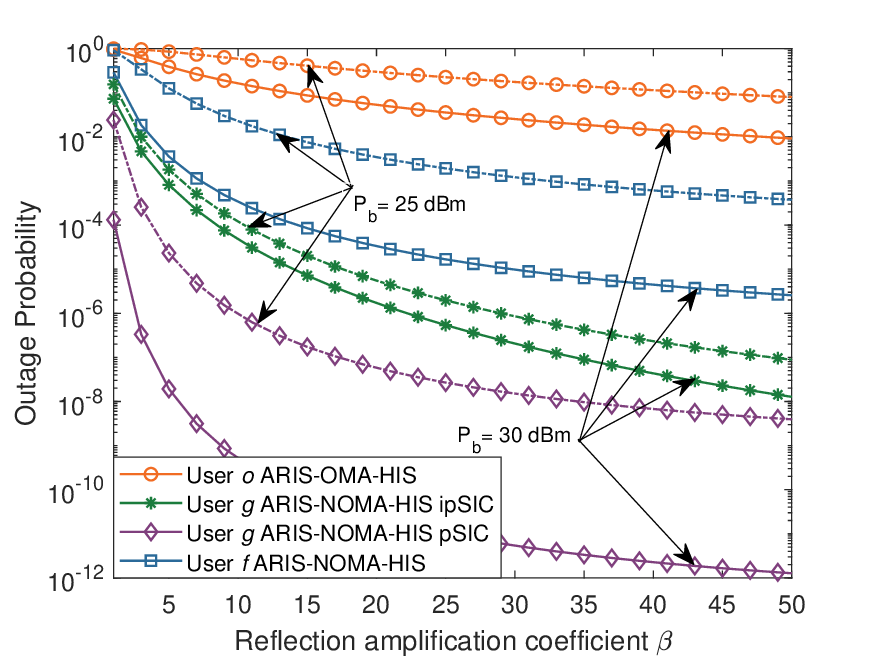}
		\caption{Outage probability versus the reflection amplitude factors, with $L=5$, $K=3$, $g=3$, $f=2$, $m=0.7$, and $R_g=R_f=1.5$ BPCU.}
		\label{ARIS_NOMA_HI_diff_beta}
	\end{center}
\end{figure}
	
\begin{figure}[t!]
	\begin{center}
		\includegraphics[width=2.9in,  height=2.1in]{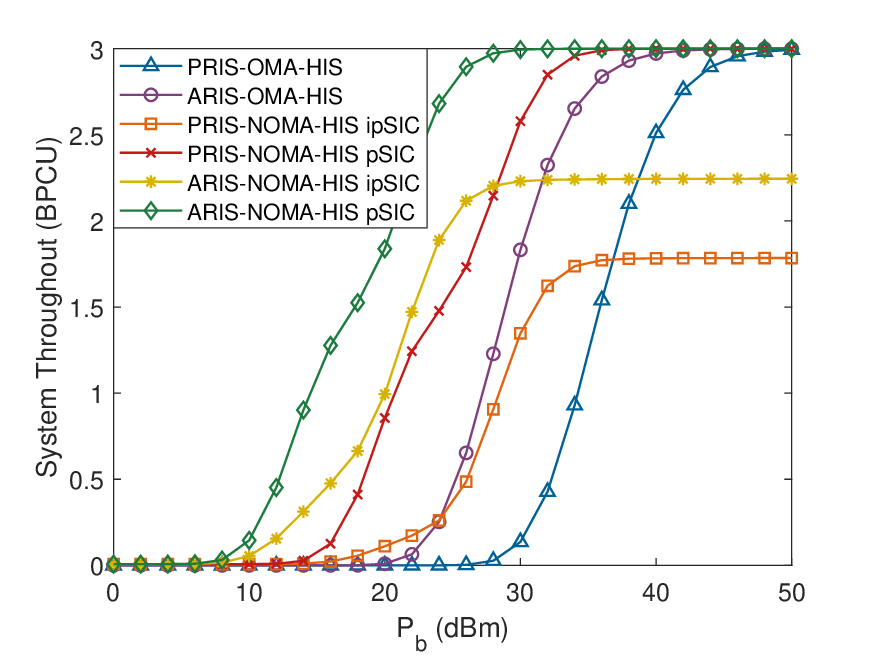}
		\caption{System throughput versus the transmit power ${P_b}$, with $L=5$, $K=3$, $g=3$, $f=2$, $m=0.7$, and $R_g=R_f=1.5$ BPCU.}
		\label{ARIS_NOMA_HI_throughput}
	\end{center}
\end{figure}
	
Fig. \ref{ARIS_NOMA_HI_diff_beta} plots the outage probability of ARIS-NOMA/OMA-HIS networks versus the different reflection amplitude factors $\beta$ with $L=5$, $K=3$, $g=3$, $f=2$, $m=0.7$ and $R_g=R_f=1.5$ BPCU. It can be seen that as $\beta$ increases, the outage performance gain gradually reaches saturation. This is because that as $\beta$ tends to infinity, the outage probability approaches a particular value. Moreover, the gain in outage probability for the ARIS-NOMA-HIS networks with pSIC is more pronounced when $P_b$ is increased from 25 dBm to 30 dBm.
Additionally, Fig. \ref{ARIS_NOMA_HI_throughput} plots the system throughput versus the transmit power ${P_b}$ in delay-limited case with $L=5$, $K=3$, $g=3$, $f=2$, $m=0.7$ and $R_g=R_f=1.5$ BPCU. The system throughput of ARIS-NOMA-HIS is plotted on the basis of ${\left(\ref{the definition of throughput}\right)}$.
As can be observed that ARIS-NOMA-HIS is able to reach throughput ceiling earlier than PRIS-NOMA-HIS. The throughput ceiling of ARIS/PRIS-NOMA-HIS with ipSIC is lower than that with pSIC because of the residual interference from ipSIC.
The system throughput of ARIS-NOMA-HIS is far more favourable than that of PRIS-NOMA-HIS and ARIS-OMA-HIS. The major contributing factors is that ARIS-NOMA-HIS networks have better spectral efficiency and the capability to offer service for multiple users.
	
\begin{figure}[t!]
	\begin{center}
		\includegraphics[width=2.9in,  height=2.1in]{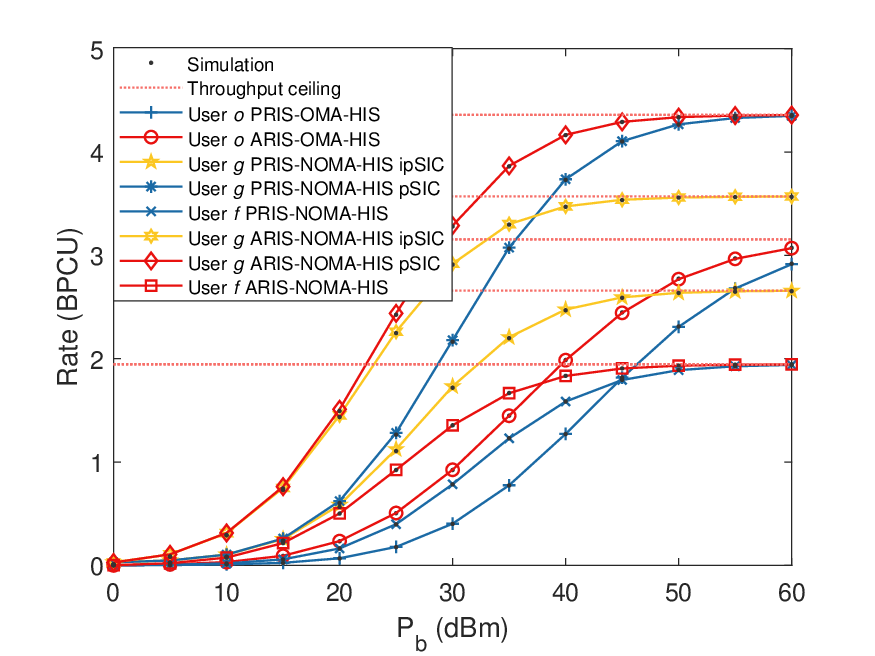}
		\caption{Rate versus the transmit power ${P_b}$, with $L=2$, $K=3$, $g=3$, $f=2$, $m=0.7$, $\beta =5$ and $R_g=R_f=1.5$ BPCU.}
		\label{ARIS_NOMA_HI_erg}
	\end{center}
\end{figure}
	
\begin{figure}[t!]
	\begin{center}
		\includegraphics[width=2.9in,  height=2.1in]{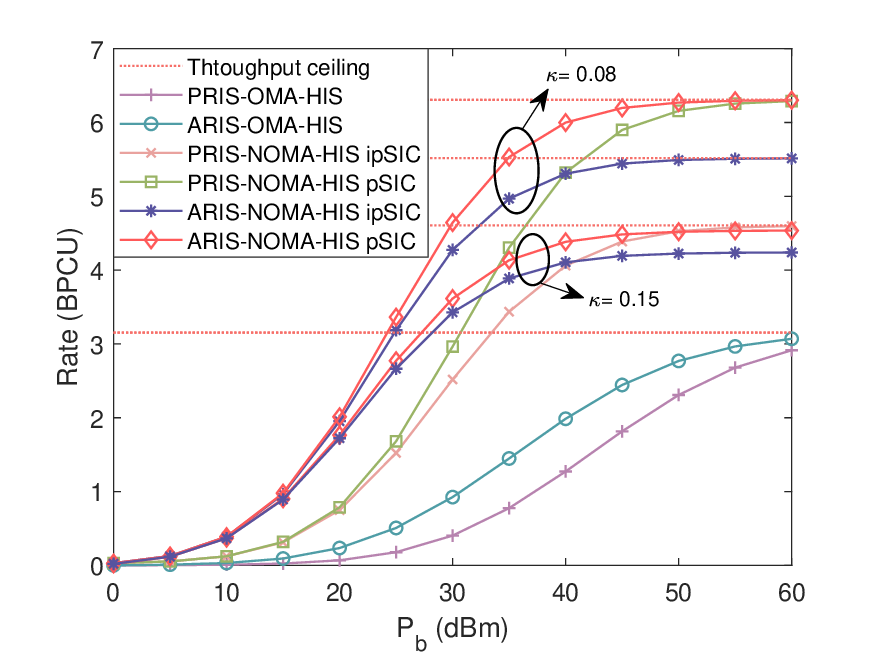}
		\caption{Rate versus the transmit power ${P_b}$, with $L=2$, $K=3$, $g=3$, $f=2$, $m=0.7$, and $R_g=R_f=1.5$ BPCU.}
		\label{ARIS_NOMA_HI_erg_dt}
	\end{center}
\end{figure}
	
Fig. \ref{ARIS_NOMA_HI_erg} plots the ergodic data rates versus the transmit power ${P_b}$ with $L=2$, $K=3$, $g=3$, $f=2$, $m=0.7$, $\beta =5$ and $R_g=R_f=1.5$ BPCU.
We can observe that the ergodic data rate for the ARIS-NOMA-HIS networks is still higher than that for PRIS-NOMA-HIS. Future more, ARIS-NOMA-HIS with ipSIC achieves a weaker rate than that with pSIC. As $P_b$ increases, the rate for user $f$ gradually reaches the upper bound and this is in alignment with the findings in \textbf{Remark \ref{Remark4:the throughput ceiling of the user g}} and \textbf{Remark \ref{Remark5:the throughput ceiling of the user f}}.
Fig. \ref{ARIS_NOMA_HI_erg_dt} plots the ergodic data rates versus the transmit power ${P_b}$ in delay-tolerate schemes with $L=2$, $K=3$, $g=3$, $f=2$, $m=0.7$ and $R_g=R_f=1.5$ BPCU. The influence of changing the HIS factor $\kappa$ on ergodic data rate of system was plotted, which basing on ${\left(\ref{the definition of delay-tolerate}\right)}$.
We can observe that as $\kappa$ increases, i.e., from 0.08 to 0.15, the ergodic performance degradation in ergodic data rate increases, and is more pronounced in ARIS/PRIS-NOMA-HIS with pSIC. At the same time, the ergodic data rata reached an upper ceiling due to the effects of HIS. This demonstrates the importance of accurately modeling HIS values when evaluating ARIS-NOMA networks. Moreover, the residual interference from ipSIC has no significant effect on ARIS-NOMA-HIS networks when $P_b$ is lower.

\begin{figure}[htbp]
	\centering
		
	\subfigure[Delay-tolerant mode]{
		\begin{minipage}[t]{0.5\linewidth} 
			\centering
			\includegraphics[width=0.9\textwidth,height=0.8\textwidth]{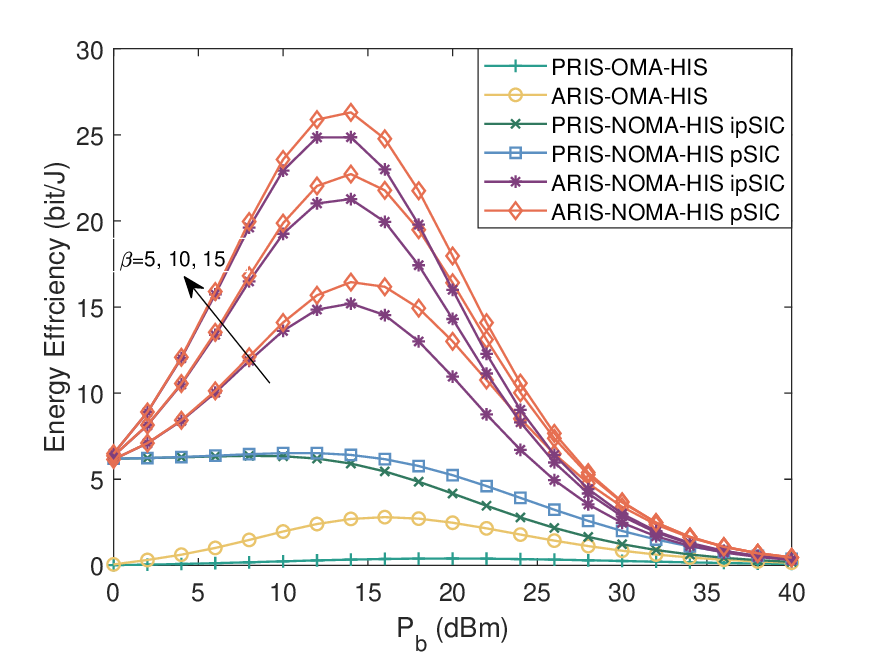}
			\label{ARIS_NOMA_HI_ee_dt}
		\end{minipage}%
	}%
	\subfigure[Delay-limited mode]{
		\begin{minipage}[t]{0.5\linewidth} 
			\centering
			\includegraphics[width=0.9\textwidth,height=0.8\textwidth]{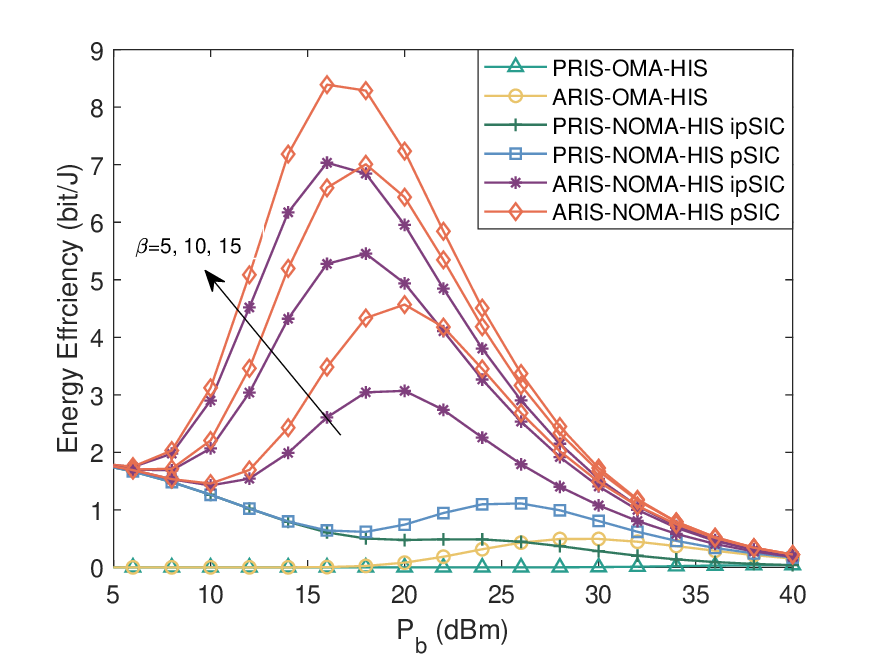}
			\label{ARIS_NOMA_HI_ee_dl}
		\end{minipage}%
	}%
		
	\centering
	\caption{Energy efficiency in delay-tolerant and delay-limited transmission mode, with $L=2$, $K=3$, $g=3$, $f=2$, $m=0.5$, and $R_g=R_f=1.5$ BPCU.}
\end{figure}
	
Fig. \ref{ARIS_NOMA_HI_ee_dt} and Fig. \ref{ARIS_NOMA_HI_ee_dl} plot the energy efficiency curves in delay-tolerant and delay-limited transmission modes with $L=2$, $K=3$, $g=3$, $f=2$, $m=0.5$ and $R_g=R_f=1.5$ BPCU, respectively.
Fig. \ref{ARIS_NOMA_HI_ee_dt} shows the delay-tolerant transmission model of ARIS/PRIS-NOMA-HIS networks. It can be observed that ARIS-OMA-HIS and PRIS-OMA-HIS have lower energy efficiency and that of the ARIS-NOMA-HIS networks is significantly improved compared to PRIS-NOMA-HIS.
As can be seen the system's energy efficiency begins to deteriorate at $P_b > 15$ dBm, This can be inferred by the fact that  the total power consumption of the system becomes larger as $P_b$ increases. Furthermore, the energy efficiency of the ARIS-NOMA-HIS dramatically increases as the growth of the reflection amplification $\beta$.
This conclusions can also be applied in delay-limited schemes. The difference is that the energy efficiency of ARIS-NOMA-HIS with ipSIC is significantly worse than that with pSIC. This demonstrates how delay-limited systems are substantially more sensitive to residual interference from ipSIC than delay-tolerant ones. As the weak trend of ARIS-NOMA-HIS networks indicates its stability in transmitting information.

\section{Conclusion}\label{Conclusion}
This paper has investigated the performance of ARIS-NOMA-HIS networks over the cascade Nakagami-m fading channels. More specifically, we derived the outage probability and ergodic data rate of ARIS-NOMA-HIS networks with ipSIC/pSIC. On the basis of the asymptotic results, the diversity orders and multiplexing gains for user $g$ and user $f$ are acquired. The results indicated that the outage performance of ARIS-NOMA-HIS outperformed PRIS-NOMA-HIS, ARIS-OMA-HIS, and other traditional collaborative relays. Based on the analysis, the diversity order is related to reflection elements number and the order of the channel.
Numerical results revealed that the outage probability and ergodic data rate of user $g$ in ARIS-NOMA-HIS networks was significantly better than that of orthogonal users at high SNRs. We studied the system throughput and energy efficiency of ARIS-NOMA-HIS in both delay-limited and delay-tolerant schemes. The single antenna and perfect CSI setup in this paper leads to an overestimation in performance of ARIS-NOMA-HIS networks, so our future work will consider the influence of imperfect CIS and take into account more complex assumptions.

\appendices
\section*{Appendix~A: Proof of Theorem \ref{Theorem1:the OP of user g}} \label{Appendix:A}
\renewcommand{\theequation}{A.\arabic{equation}}
\setcounter{equation}{0}
	
Based on  the definitions of outage probability for user $g$, the proof processes start by substituting  \eqref{The SINR of user g to detect user f} and \eqref{The SINR of user g} into \eqref{the outage probability of user g}, and then using simple arithmetic operations, the corresponding outage probability for ARIS-NOMA-HIS networks can be calculated as
\begin{small}
	\begin{align}\label{the Appendix for OP expression of user g}
		P_{ARIS,g}^{ipSIC} =& {\rm{Pr}}\left[ {0 < {{\left| {{\bf{h}}_{rg}^H{\bf{\Theta }} {{\bf{h}}_{br}}} \right|}^2} < {\varsigma _g}{\beta ^{ - 1}}\left( {\varpi {P_b}{{\left| {{h_{\rm{I}}}} \right|}^2} + {n_g}} \right)} \right]\nonumber\\
		=& {\rm{Pr}}\left[ {0 < \underbrace {\left| {\sum\limits_{l = 1}^L {h_{br}^lh_{rg}^l} } \right|}_Y < \sqrt {{\zeta _g}{\beta ^{ - 1}}\left( {\varpi {P_b}\underbrace {{{\left| {{h_{\rm{I}}}} \right|}^2}}_X + {n_g}} \right)} } \right]\nonumber\\
		=& \int_0^\infty  {{f_X}\left( x \right){F_Y}\left( {\sqrt {{\zeta _g}{\beta ^{ - 1}}\left( {\varpi {P_b}x + {n_g}} \right)} } \right)} dx,
	\end{align}
\end{small}
where $\varpi  = \xi  = 1$, ${n_g} = \xi \beta {N_{tn}}L{\omega _{rg}} + {\sigma ^2}$, ${\zeta _g} = d_{br}^\alpha d_{rg}^\alpha {\varsigma _g}$. Note that the equation \eqref{the Appendix for OP expression of user g} is derived with the assistance of coherent phase shifting, which can maximize the received performance of the desired users. To compute the equation \eqref{the Appendix for OP expression of user g}, we should first solve the PDF and CDF of random variables $X = {\left| {{h_{\rm{I}}}} \right|^2}$ and $Y= \sum\nolimits_{l = 1}^L {\left| {h_{br}^lh_{rg}^l} \right|} $, respectively. Since ${h_{I}}$ is assumed to obey the Nakagami-$m$ distribution with the parameter of $m_g$, ${\left| {h_{I}} \right|^2}$ is subjected to the Gamma distribution having the parameter of $\left( {{{m_g}},{{m_g}}} \right)$ and the corresponding PDF is given by
\begin{align}\label{the Appendix for PDF of X}
	{f_X}\left( x \right) = \frac{{{m_g}^{{m_g}}{x^{{m_g} - 1}}}}{{\Gamma \left( {{m_g}} \right)}}{e^{ - {m_g}x}}.
\end{align}
It may be seen that the CDF of $Y$ cannot be derived in a straightforward way. However, we can make use of the Laguerre polynomials to provide the approximated CDF. Let ${Y_l} = \left| {h_{br}^lh_{rg}^l} \right|$, the expectation and variance of ${Y_l}$ can be respectively given by
$ {\mu _g} ={\mathbb{E}}\left( {{Y_l}} \right) = \frac{{\Gamma \left( {{m_r} + \frac{1}{2}} \right)\Gamma \left( {{m_g} + \frac{1}{2}} \right)}}{{\Gamma \left( {{m_r}} \right)\Gamma \left( {{m_g}} \right)\sqrt {{m_r}{m_g}} }} $ and ${\Omega _g} = {\mathbb{D}}\left( {{Y_l}} \right)= 1 - \frac{1}{{{m_r}{m_g}}}{\left[ {\frac{{\Gamma \left( {{m_r} + \frac{1}{2}} \right)\Gamma \left( {{m_g} + \frac{1}{2}} \right)}}{{\Gamma \left( {{m_r}} \right)\Gamma \left( {{m_g}} \right)}}} \right]^2}$, where $m_r$ denote the Nakagami-$m$ fading parameter from BS to ARIS.
Conditioned that there is no channel sorting, the PDF of $Y$ can be approximated as ${f_Y}\left( y \right)  \approx  \frac{{{y^{{b_g}}}{e^{ - \frac{y}{{{c_g}}}}}}}{{c_g^{{b_g} + 1}\Gamma \left( {{b_g} + 1} \right)}}$ \cite[Eq. (2.76)]{Primak2004}, where ${b_g} = \frac{{L{{\left[ {{\mathbb{E}}\left( {{Y_l}} \right)} \right]}^2}}}{{{\mathbb{D}}\left( {{Y_l}} \right)}} - 1$, ${c_g} = \frac{{{\mathbb{D}}\left( {{Y_l}} \right)}}{{{\mathbb{E}}\left( {{Y_l}} \right)}}$. The integration operations of ${f_Y}\left( y \right) $ are carried out, and then its CDF can be further derived as ${F_Y}\left( y \right) \approx \frac{1}{{\Gamma \left( {{b_g} + 1} \right)}}\gamma \left( {{b_g} + 1,\frac{y}{{{c_g}}}} \right)$. As a further advance, the CDF of $Y$ with channel sorting can be expressed as
\begin{align}\label{the Appendix: the approximated CDF of Y}
	{F_Y}\left( y \right) \approx & \frac{{K!}}{{\left( {K - g} \right)!\left( {g - 1} \right)!}}\sum\limits_{k = 0}^{K - g} {{
			{K - g}  \choose
			k  }}  \nonumber \\
	&\times \frac{{{{\left( { - 1} \right)}^k}}}{{g + k}}{\left[ {\frac{1}{{\Gamma \left( {{b_g} + 1} \right)}}\gamma \left( {{b_g} + 1,\frac{y}{{{c_g}}}} \right)} \right]^{g + k}}.
\end{align}
By substituting \eqref{the Appendix for PDF of X} and \eqref{the Appendix: the approximated CDF of Y} into \eqref{the Appendix for OP expression of user g}, the outage probability of user $g$ for ARIS-NOMA-HIS is approximated as
\begin{small}
	\begin{align}\label{the Appendix: the approximated OP of user g}
		P_{ARIS,g}^{ipSIC} =& \frac{{K!}}{{\left( {K - g} \right)!\left( {g - 1} \right)!}}\sum\limits_{k = 0}^{K - g} {{
				{K - g}  \choose
				k  }}\frac{{{{\left( { - 1} \right)}^k}{m_g}^{{m_g}}}}{{\left( {g + k} \right){\psi_g}}}\nonumber \\
		& \times {\int_0^\infty  {\frac{{{x^{{m_g} - 1}}}}{{{e^{{m_g}x}}}}\left[ {\gamma \left( {{b_g} + 1,\frac{{\sqrt {{\varphi _g}x + {{\tilde \vartheta }_g}} }}{{{c_g}\sqrt \beta  }}} \right)} \right]} ^{g + k}}dx.
	\end{align}
\end{small}
Applying the Gauss-Laguerre quadrature \cite[Eq. (8.6.5)]{Hildebrand1987introduction} into the above equation, we will acquire \eqref{the OP with ipSIC of user g ARIS-NOMA-HIS}. The proof is completed.
	
\section*{Appendix~B: Proof of Theorem \ref{Theorem2:the OP of user f}} \label{Appendix:B}
\renewcommand{\theequation}{B.\arabic{equation}}
\setcounter{equation}{0}
	
By substituting \eqref{The SINR of user f} into \eqref{the outage probability of user f}, the outage probability for user $f$ is evaluated by
\begin{align}\label{the Appendix for OP expression of user f}
	P_{ARIS,f}^{HIS} &= {\rm{Pr}}\left[ {{{\left| {{\bf{h}}_{rf}^H{\bf{\Theta }} {{\bf{h}}_{br}}} \right|}^2} < {\varsigma _f}{\beta ^{ - 1}}\left( {\xi \beta {N_{tn}}L{\omega _{rf}} + {\sigma ^2}} \right)} \right] \nonumber\\
	& = {\rm{Pr}}\left[ {\underbrace {\left| {\sum\limits_{l = 1}^L {h_{br}^lh_{rf}^l} } \right|}_W < \sqrt {d_{br}^\alpha d_{rf}^\alpha {\varsigma _f}{\beta ^{ - 1}}{n_f}} } \right]\nonumber\\
	& = {F_W}\left( {\sqrt {d_{br}^\alpha d_{rf}^\alpha {\varsigma _f}{\beta ^{ - 1}}{n_f}} } \right),
\end{align}
where $\xi  = 1$, ${n_f} = \xi \beta {N_{tn}}L{\omega _{rf}} + {\sigma ^2}$. For the same reason as \eqref{the Appendix for PDF of X} and \eqref{the Appendix: the approximated CDF of Y}, we can derive the CDF of $W$ with respect to the sorted as
\begin{align}\label{the Appendix: the approximated CDF of W}
	{F_W}\left( w \right) \approx & \frac{{K!}}{{\left( {K - f} \right)!\left( {f - 1} \right)!}}\sum\limits_{k = 0}^{K - f} {{
			{K - f}  \choose
			k  }}  \nonumber \\
	&\times \frac{{{{\left( { - 1} \right)}^k}}}{{f + k}}{\left[ {\frac{1}{{\Gamma \left( {{b_f} + 1} \right)}}\gamma \left( {{b_f} + 1,\frac{w}{{{c_f}}}} \right)} \right]^{f + k}}.
\end{align}
By substituting \eqref{the Appendix: the approximated CDF of W} into \eqref{the Appendix for OP expression of user f}, we will acquire \eqref{the OP of user f ARIS-NOMA-HIS}. The proof is completed.
	
\section*{Appendix~C: Proof of Theorem \ref{Theorem4:the ER of user g}} \label{Appendix:C}
\renewcommand{\theequation}{C.\arabic{equation}}
\setcounter{equation}{0}
	
By substituting \eqref{The SINR of user g} into \eqref{the definition of ergodic rate} and further exploiting the coherent phase shifting scheme, the ergodic data rate for user $g$ of ARIS-NOMA-HIS is calculated by
\begin{align}\label{the Appendix for ER expression of user g}
	R_{ARIS,g}^{erg} &= \mathbb{E}\left[ {\log \left( {1 + {X_1}} \right)} \right] \nonumber\\
	&= \frac{{\beta {P_b}{a_g}}}{{\ln 2}}\int_0^{\frac{1}{{\beta {P_b}\left( {\kappa _b^2 + \kappa _g^2} \right)}}} {\frac{{1 - {F_{{X_1}}}\left( y \right)}}{{1 + \beta {P_b}{a_g}y}}} dy,
\end{align}
and
\begin{align}
	{X_1} = \frac{{\beta {P_b}{a_g}d_{br}^{ - \alpha }d_{rg}^{ - \alpha }{{\left| {\sum\limits_{l = 1}^L {h_{br}^lh_{rg}^l} } \right|}^2}}}{{\beta {P_b}d_{br}^{ - \alpha }d_{rg}^{ - \alpha }{{\underbrace {\left| {\sum\limits_{l = 1}^L {h_{br}^lh_{rg}^l} } \right|}_Y}^2}{\chi _g} + \varpi {P_b}\underbrace {{{\left| {{h_{\rm{I}}}} \right|}^2}}_X + {n_g}}}.\nonumber
\end{align}
The CDF of ${X_1}$ can be written specifically as
\begin{small}
	\begin{align}\label{the expression of FX1-1}
		{F_{{X_1}}}\left( y \right) &= \Pr \left[ {\left| {\sum\limits_{l = 1}^L {h_{br}^lh_{rg}^l} } \right| < \sqrt {\frac{{yd_{br}^\alpha d_{rg}^\alpha \left( {\varpi {P_b}{{\left| {{h_{\rm{I}}}} \right|}^2} + {n_g}} \right)}}{{1 - y\beta {P_b}{\chi _g}}}} } \right]\nonumber\\
		&= \int_0^\infty  {{f_X}\left( x \right)} {F_Y}\left( {\sqrt {\frac{{yd_{br}^\alpha d_{rg}^\alpha \left( {\varpi {P_b}x + {n_g}} \right)}}{{1 - y\beta {P_b}\left( {\kappa _b^2 + \kappa _g^2} \right)}}} } \right)dx.
	\end{align}
\end{small}
By substituting \eqref{the Appendix for PDF of X} and \eqref{the Appendix: the approximated CDF of Y} into \eqref{the expression of FX1-1} and applying the Gauss-Laguerre integral, ${F_{{X_1}}}\left( y \right)$ can be further written as
\begin{small}
	\begin{align}\label{the expression of FX1}
		{F_{{X_1}}}\left( y \right) =& \sum\limits_{u = 1}^U {{H_u}} x_u^{{m_g} - 1}{\Psi _g}\sum\limits_{k = 0}^{K - g} {{
				{K - g}  \choose
				k  }\frac{{{{\left( { - 1} \right)}^k}}}{{\left( {g + k} \right)\Gamma \left( {{m_g}} \right)}}}  \nonumber\\
		& \times {\left[ {\frac{1}{{\left[ {\Gamma \left( {{b_g} + 1} \right)} \right]}}\gamma \left( {{b_g} + 1,\frac{{\sqrt {y{\varpi _{g1}}{x_u} + y{\varpi _{g2}}} }}{{{c_g}\sqrt {{m_g}\left( {1 - y{\tau _g}} \right)} }}} \right)} \right]^{g + k}}.
	\end{align}
\end{small}
In order to adequately simplify the results obtained, we use the following Gauss-Chebyshev integral formula
\begin{align}\label{Gauss-Chebyshev integral}
	\int_0^b {f\left( x \right)} dx \approx \frac{{\pi b}}{{2N}}\sum\limits_{n = 1}^N {f\left( {\frac{{\left( {{x_n}{\rm{ + }}1} \right)b}}{2}} \right)\sqrt {1 - {x_n}^2} } .
\end{align}

By substituting \eqref{the expression of FX1} into \eqref{the Appendix for ER expression of user g} and applying \eqref{Gauss-Chebyshev integral}, we will acquire \eqref{the ergodic rate with ipSIC of user g in ARIS}. The proof is completed.
	
\section*{Appendix~D: Proof of Theorem \ref{Theorem5:the ER of user f}} \label{Appendix:D}
\renewcommand{\theequation}{D.\arabic{equation}}
\setcounter{equation}{0}
	
By substituting \eqref{The SINR of user f} into \eqref{the definition of ergodic rate}, the ergodic data rate for user $f$ of ARIS-NOMA-HIS networks can be calculated by
\begin{align}\label{the Appendix for ER expression of user f}
	R_{ARIS,f}^{erg} &= \mathbb{E}\left[ {\log \left( {1 + {X_2}} \right)} \right] \nonumber\\
	&= \frac{{\beta {P_b}{a_f}}}{{\ln 2}}\int_0^{\frac{1}{{\beta {P_b}\left( {{a_g} + \kappa _b^2 + \kappa _f^2} \right)}}} {\frac{{1 - {F_{{X_2}}}\left( y \right)}}{{1 + \beta {P_b}{a_f}y}}} dy,
\end{align}
and
\begin{align}
	{X_2} = \frac{{\beta {P_b}d_{br}^{ - \alpha }d_{rf}^{ - \alpha }{{\left| {\sum\limits_{l = 1}^L {h_{br}^lh_{rf}^l} } \right|}^2}{a_f}}}{{\beta {P_b}d_{br}^{ - \alpha }d_{rf}^{ - \alpha }{{\left| {\sum\limits_{l = 1}^L {h_{br}^lh_{rf}^l} } \right|}^2}{\chi _f} + \xi \beta {N_{tn}}L{\omega _{rf}} + {\sigma ^2}}}.\nonumber
\end{align}
The CDF of ${X_2}$ can be written specifically as
\begin{small}
	\begin{align}\label{the expression of FX2-1}
		{F_{{X_2}}}\left( y \right) &= \Pr \left[ {\left| {\sum\limits_{l = 1}^L {h_{br}^lh_{rf}^l} } \right| < \sqrt {\frac{{yd_{br}^\alpha d_{rf}^\alpha \left( {\xi \beta {N_{tn}}L{\omega _{rf}} + {\sigma ^2}} \right)}}{{1 - y\beta {P_b}\left( {{a_g} + \kappa _b^2 + \kappa _f^2} \right)}}} } \right]\nonumber\\
		&= {F_Y}\left( {\sqrt {\frac{{yd_{br}^\alpha d_{rf}^\alpha \left( {\xi \beta {N_{tn}}L{\omega _{rf}} + {\sigma ^2}} \right)}}{{1 - y\beta {P_b}\left( {{a_g} + \kappa _b^2 + \kappa _f^2} \right)}}} } \right).
	\end{align}
\end{small}
By substituting \eqref{the Appendix for PDF of X} and \eqref{the Appendix: the approximated CDF of Y} into \eqref{the expression of FX2-1} and applying the Gauss-Laguerre integral, ${F_{{X_2}}}\left( y \right)$ can be rewritten as
\begin{small}
	\begin{align}\label{the expression of FX2}
		{F_{{X_2}}}\left( y \right) &= \frac{{K!}}{{\left( {K - f} \right)!\left( {f - 1} \right)!}}\sum\limits_{k = 0}^{K - f} {{
				{K - f}  \choose
				k  }\frac{{{{\left( { - 1} \right)}^k}}}{{f + k}}}   \nonumber\\
		& \times {\left[ {\frac{1}{{\Gamma \left( {{b_f} + 1} \right)}}\gamma \left( {{b_f} + 1,\frac{{\sqrt {y{\varpi _{f1}} + y{\varpi _{f2}}} }}{{{c_f}\sqrt {1 - y{\tau _f}} }}} \right)} \right]^{f + k}}.
	\end{align}
\end{small}
	
By substituting \eqref{the expression of FX2} into \eqref{the Appendix for ER expression of user f}, we will acquire \eqref{the ergodic rate of user f in ARIS}. The proof is completed.
	
\bibliographystyle{IEEEtran}
\bibliography{mybib}

\end{document}